\documentclass[11pt]{article}

\usepackage{titlesec}

\usepackage{caption}
\newcommand{\appendixformat}{%
	\titleformat{\section}
	{\normalfont\Large\bfseries}
	{Appendix \thesection}
	{1em}{}
}

\usepackage[T1]{fontenc}
\usepackage[utf8]{inputenc}
\usepackage{lmodern}

\usepackage{amsmath, amssymb, amsfonts}
\usepackage{mathtools}
\usepackage{bm}
\usepackage{bbm}

\usepackage[margin=0.75in]{geometry}

\usepackage{graphicx}
\usepackage{caption}
\usepackage{subcaption}

\usepackage{amsthm}
\theoremstyle{plain}
\newtheorem{theorem}{Theorem}[section]
\newtheorem{lemma}[theorem]{Lemma}
\newtheorem{proposition}[theorem]{Proposition}
\newtheorem{corollary}[theorem]{Corollary} 
\theoremstyle{definition}

\theoremstyle{remark}
\newtheorem{remark}[theorem]{Remark}

\usepackage{algorithm}
\usepackage{algpseudocode}

\usepackage{booktabs}
\usepackage{array}
\usepackage{multirow}

\usepackage{xcolor}
\definecolor{blueColor}{rgb}{0,0.2,0.6}

\usepackage[colorlinks=true, linkcolor=blueColor, citecolor=blueColor, urlcolor=blueColor]{hyperref}
\usepackage{siunitx}

\def\vx{{\bf{x}}}  
\def\vX{{\bf{X}}}
\def\vy{{\bf{y}}} 
\def\vv{{\bf{v}}} 
\def\vB{{\bf{B}}} 
\def\bR{\mathbb{R}}

\def\bE{\mathbb{E}}
\def\cU{\mathcal{U}}
\def\cO{\mathcal{O}}
\def\freq{\mathsf{f}}
\def\sfE{\mathsf{E}}
\def\cF{\mathcal{F}}
\def\Re{\mathrm{Re}}

\allowdisplaybreaks

\title{Compartmental-reaction diffusion framework for microscale dynamics of extracellular serotonin in brain tissue}

\author{Merlin Pelz\thanks{Department of Mathematics, University of Minnesota, Minneapolis, MN, USA
		(mpelz@umn.edu, ghandy@umn.edu).}
	\and Skirmantas Janu\v{s}onis\thanks{Department of Psychological \& Brain Sciences, University of California Santa Barbara, Santa Barbara, CA, USA (janusonis@ucsb.edu).}
	\and Gregory Handy\footnotemark[1]}

\begin{document}

\maketitle

\begin{abstract}
Serotonin (5-hydroxytryptamine) is a major neurotransmitter whose release from densely distributed serotonergic varicosities shapes plasticity and network integration throughout the brain, yet its extracellular dynamics remain poorly understood due to the sub-micrometer and millisecond scales involved. We develop a mathematical framework that captures the coupled reaction-diffusion processes governing serotonin signaling in realistic tissue microenvironments. Formulating a two-dimensional compartmental-reaction diffusion system, we use strong localized perturbation theory to derive an asymptotically equivalent set of nonlinear integro-ODEs that preserve diffusive coupling while enabling efficient computation. We analyze period-averaged steady states, establish bounds using Jensen’s inequality, obtain closed-form spike maxima and minima, and implement a fast marching-scheme solver based on sum-of-exponentials kernels. These mathematical results provide quantitative insight into how firing frequency, varicosity geometry, and uptake kinetics shape extracellular serotonin. The model reveals that varicosities form diffusively coupled microdomains capable of generating spatial “serotonin reservoirs,” clarifies aspects of local versus volume transmission, and yields predictions relevant to interpreting high-resolution serotonin imaging and the actions of selective serotonin-reuptake inhibitors.
\end{abstract}

\section{Introduction}

Serotonin (5-hydroxytryptamine, 5-HT) is a major neurotransmitter in the central nervous systems of all vertebrates, from ancient fish lineages to mammals ~\cite{jacobs1992structure,awasthi_comprehensive_2021,moroz_neural_2021,Janusonis2025}. It is delivered to the entire brain by axons, thin and long structures that can extend from serotonergic cell bodies to any brain region. Within these regions, serotonergic axons typically travel in unique, random-walk-like trajectories~\cite{Janusonis2020,mays2023experimental} and can aggregate at extremely high densities~\cite{Linley2013,awasthi_comprehensive_2021,mays2023experimental} (Figure 1A). Because of these properties, serotonergic axons are often referred to as serotonergic fibers. Geometrically, they often resemble fibers in a cotton ball, where each fiber has an unpredictable, unique, and continuously changing orientation. Brain networks that allow topological, graph-like descriptions (“neuron A connects to neuron B”)~\cite{Lynn2019} are physically embedded in these serotonergic fiber meshworks. Both systems are axon-based, but the latter one cannot be described topologically, without a reference to the physical space. The time-scale of morphological fiber growth is much longer than that of serotonin signaling, and hence the fiber arrangement can be treated as being in quasi-equilibrium throughout this work.

The fundamental role of serotonin signaling remains elusive, and serotonergic fibers have no analogs in current artificial neural networks (ANNs)~\cite{Lee2022b}. An early review has noted that ``although serotonin has been implicated in a vast array of physiological and behavioral processes in vertebrates, it appears to be essential for none of them''~\cite{jacobs1992structure}. However, recent research suggests that serotonin signaling fundamentally supports neuroplasticity~\cite{Lesch2012,Teissier2017}, which may explain its involvement in virtually all neuroprocesses. Furthermore, dysfunctional serotonin signaling has been associated with complex mental disorders such as depression, anxiety, obsessive-compulsive disorder, and post-traumatic stress disorder~\cite{pourhamzeh_roles_2022,Page2024,Faraone2025}; all of them have behavioral aspects that can be viewed as maladaptive persistence. In particular, selective serotonin-reuptake inhibitors (SSRIs), widely prescribed antidepressants, directly alter serotonergic signaling, likely enhancing neuroplasticity~\cite{Page2024}.

Serotonergic fiber densities vary across brain regions~\cite{Awasthi2021}, and serotonergic neurons can express a wide range of transcriptional programs~\cite{Okaty2019, Ren2019}. However, serotonergic fibers share many similarities. A typical serotonergic fiber is thread-like, with a diameter that may approach the limit of optical resolution (around \SI{200}{\nano\meter}), but it also contains short dilated segments, known as varicosities~\cite{Maddaloni2017,Gianni2023a}. Varicosities can have a much larger diameter, often producing a ``beads on a string'' impression in light microscopy (Figure~\ref{fig:serotoninsystem}A-C)~\cite{Maddaloni2017,hingorani2022high,haiman2023qualitative}. They can measure up to \SI{5}{\micro\meter} across in cell culture preparations~\cite{Benzekhroufa2009a,hingorani2022high} but are typically smaller in natural brain tissue, around \SI{1}{\micro\meter} or less in diameter~\cite{Maddaloni2017,haiman2023qualitative}. Varicosities have been primarily studied in fixed tissue; in the living brain, they are likely to be dynamic and change in size during brain development~\cite{Maddaloni2017} and in response to external perturbations, such as antidepressants~\cite{Nazzi2019}. 

Serotonergic varicosities are considered to be the primary sites of serotonin release into the extracellular space, a pulse-like process triggered by action potentials that get initiated at the cell body and travel down the fiber~\cite{Benzekhroufa2009a}. Varicosities are also thought to be the primary sites of serotonin reuptake, whereby extracellular serotonin is continuously ``pumped'' back into the fiber by the serotonin transporter (SERT), a specialized transmembrane protein~\cite{Henke2018}. Acting together, serotonin release and reuptake (Figure~\ref{fig:serotoninsystem}D, left) create locally-specific serotonin concentrations in the extracellular space (ECS), which are thought to be in the nanomolar range~\cite{Adell2002a,Zhao2023}. These concentrations can fluctuate in time~\cite{Cooper2025}, probably on several temporal scales.  
\begin{figure}[!htbp]
    \centering
    \includegraphics[width=0.75\textwidth]{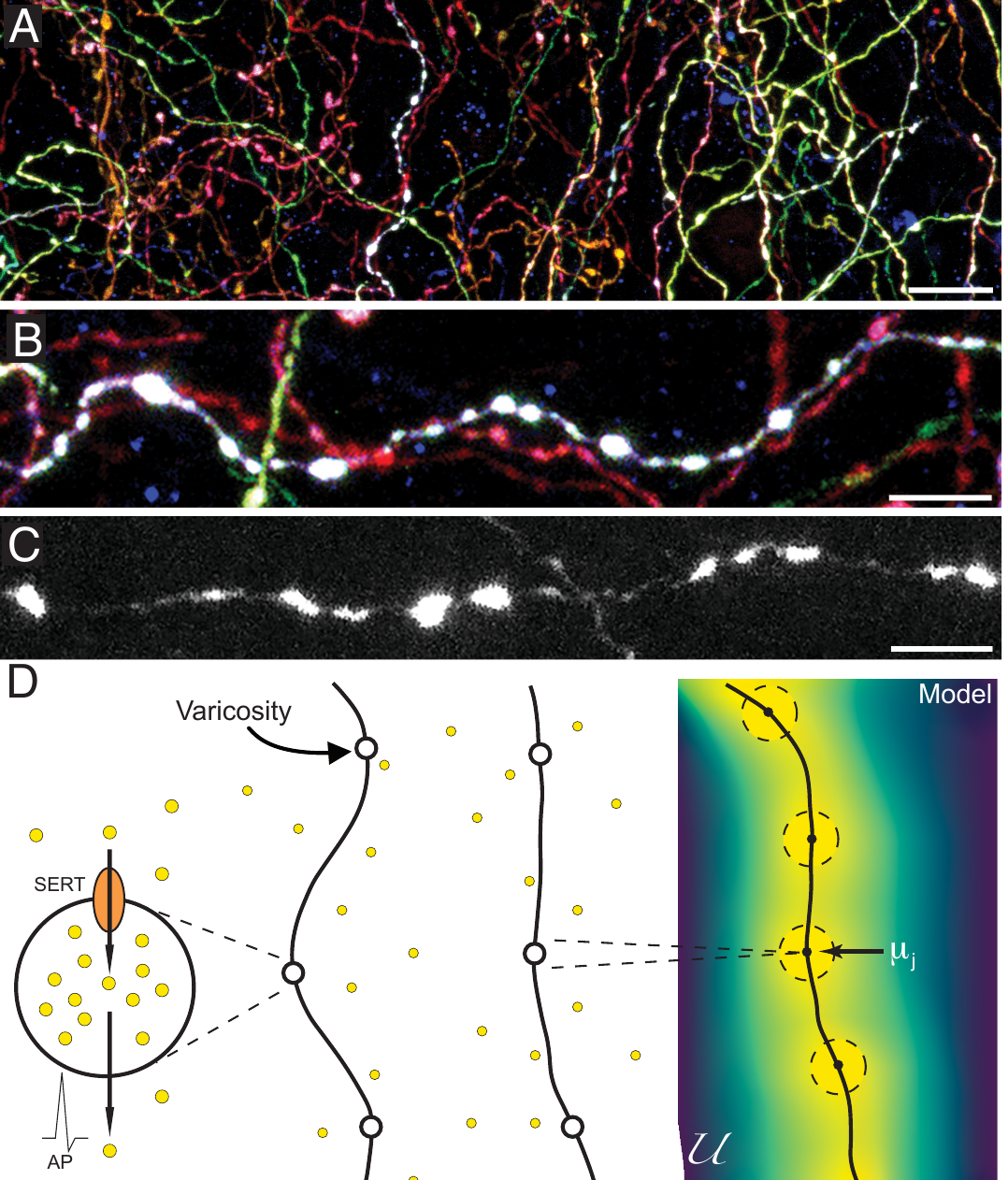}
    \caption{A: A typical meshwork of serotonergic fibers in the mouse forebrain (the hippocampus region). Individual fibers are labeled with different colors with the Brainbow transgenic technique and are shown in a maximum-intensity projection of a \SI{40}{\micro\meter}-thick section (scale bar = \SI{10}{\micro\meter}). B: A serotonergic fiber in the mouse forebrain (scale bar = \SI{2}{\micro\meter}). C: A serotonergic fiber in the forebrain of the Pacific angelshark (scale bar = \SI{2}{\micro\meter}). In A-C, the varicosities (dilated fiber segments) are clearly visible. The fibers were visualized with serotonin immunohistochemistry and imaged with confocal microscopy. These previously unpublished images were acquired in studies described in~\cite{mays2023experimental,haiman2023qualitative,Janusonis2025}. D: Schematic of the system at various scales. Left shows a typical varicosity that releases serotonin in response to action potentials (AP) and continuously takes up extracellular serotonin via SERT. Middle shows the diffusion of individual molecules (yellow circles) in the ECS between fibers and varicosities (white circles). Right shows the continuum limit of this setup, with well-mixed yellow varicosity neighborhoods and concentration gradients in the ECS shortly after the fiber has fired (blue: low concentration, yellow: high concentration).}
    \label{fig:serotoninsystem}
\end{figure}

The spatiotemporal dynamics of serotonin in the ECS may be a key mechanism that supports the brain’s capability to continuously learn and adapt while preserving global circuit integration. This balance remains a major challenge in current ANN architectures~\cite{Dohare2024}. The ECS is a geometrically complex but continuous system of narrow spaces that connects the entire brain~\cite{Hrabetova2018}. A long-standing question in neuroscience is whether the serotonin released by a varicosity diffuses away through the ECS to affect multiple cells (in a process similar to ``public broadcasting'') or whether it remains physically confined to one synaptic cleft (akin to a ``private line''). Traditionally, serotonin signaling has been considered to be of the first kind (called ``volume neurotransmission''), implying that serotonergic varicosities do not form classical synapses. However, evidence suggests that serotonin may employ both signaling modes~\cite{Papadopoulos1987,Gianni2023a}. Furthermore, a recent study has shown that the two transmission modes may be only the extremes of a gradient that is controlled by the release-reuptake balance, an inherently dynamic process~\cite{Zhang2025b}. A local serotonin concentration can be detected by all brain cells in the area (neurons, glia, endothelium) that express various subsets of serotonin receptors~\cite{Bockaert2006}. Serotonin can also enter cells and their membrane-bound organelles (e.g., mitochondria) by facilitated diffusion, but this process is rarely considered in neurotransmission~\cite{Andrews2022a}.      

A comprehensive understanding of serotonin dynamics in brain tissue requires experimental methods that support spatial analyses with sub-micrometer resolution and temporal analyses with sub-second temporal resolution. In addition, they should selectively detect serotonin in the nanomolar range (with other molecules present in the tissue). This set of requirements is very challenging and cannot be met by any of the currently available techniques, including microdialysis and voltammetry. Recently, several genetically-engineered probes (e.g., GRAB$_\text{5HT}$, iSeroSnFR) have been produced that have superior spatial and temporal resolution, with the potential to revolutionize the field~\cite{Unger2020,Wan2021a,Zhang2025b}. However, data obtained with these methods remain incomplete and require mathematical models for interpretations and quantitative predictions. In addition to theoretical rigor, these models should be sufficiently rich to capture the irreducible biological complexity.   

This present study builds on the prior reaction–diffusion (RD) theory to develop a new framework for serotonin dynamics, addressing this key gap in understanding. RD systems have been extensively used in neuroscience where the self-organization and function of systems often depend on their spatial properties~\cite{ermentrout2010mathematical}. An RD equation comprises both kinetic terms for reaction dynamics (usually of chemical species), as well as a diffusion term through which Brownian diffusion (of the same chemical species) is incorporated. RD systems consist of two or more RD equations in which the kinetic terms capture the interactions among the species at each point in the domain, typically with one or more nonlinear interactions. 

Turing's 1952 work on tissue morphogenesis provided an initial explanation to how spatial patterns and symmetry breaking can arise in biological systems and sparked a wave of research on the interplay of reactions and diffusion as a mechanism for self-organization \cite{Turing1952}. A key advance came from Gierer and Meinhardt who showed the importance of coupling a slowly diffusing activator and a rapidly diffusing inhibitor, where the activator stimulates the production of both molecules while the inhibitor suppresses their production~\cite{gierer1972theory, gierer1981generation}. When these species interact across a domain while diffusing on different timescales, their interactions can give rise to complex spatiotemporal patterns of alternating high and low concentrations, as in the model from~\cite{Granero1977}, where the inhibitor must diffuse faster to produce such effects. In mathematical neuroscience, continuous RD systems (e.g., the systems of differential equations in \cite{kopell1973plane, bard1981waves, maginu1985geometrical}) and their discretized counterparts (e.g., the model in \cite{carpio2000wave}) have been used to describe the brain as an excitable medium, revealing characteristic patterns such as traveling waves. Recent work has derived a two-component RD model for the striatum, with cholinergic interneurons as the activator and dopamine as the inhibitor, showing that dopamine traveling waves can be captured by a continuous-space framework despite the compartmentalized structure of this brain region~\cite{matityahu2023acetylcholine}.

The RD framework has recently been extended to account for diffusion between dynamically active but spatially fixed boundaries, where nonlinear reactions occur directly on domain interfaces \cite{gomez2007self, levine2005membrane}. This work has led to the development of compartmental-reaction diffusion (CRD) systems \cite{rauch2004role, iyaniwura2021synchronous, gomez2021pattern}, culminating in theories of quasi–steady-state pattern formation in one and two dimensions \cite{pelz2023emergence, pelz2023symmetry} and diffusively coupled synchronized oscillators in two dimensions \cite{pelzward2025synchronized}. CRD systems capture the natural compartmentalization of biological tissues, where dynamically active compartments with nonlinear reaction kinetics (often activator–inhibitor) are coupled through extracellular diffusion and boundaries with species-specific permeabilities. Patterns emerge when compartments settle into distinct quasi–steady-states (spatial patterns) or oscillatory states that synchronize within subgroups (spatiotemporal patterns). Importantly, CRD systems generate patterns under less restrictive conditions than classical RD models, which typically require unrealistically large diffusive timescale differences between species. While alternative approaches such as adding immobile chemical species \cite{pearson1992pattern, korvasova2015investigating} or fine-tuning kinetics \cite{pearson1989turing} have been explored, their biological generality remains uncertain.

In the present article, we consider a 2-D CRD system, where each compartment represents the neighborhood of a serotonergic varicosity and the coupling between compartments occurs through bulk diffusion in the extracellular space. In an important theoretical modification of this framework, serotonin is released into a subset of these compartments, in discrete temporal (smoothed) kicks with certain frequencies, and is absorbed by all compartments continuously in time. Hence, a mathematical limit is used to make the compartment boundaries fully permeable, so that the CRD system becomes a system of discrete diffusively coupled nonlinear ordinary differential equations for all varicosities with their neighborhoods. In this way, highly accurate approximations of the extracellular and local serotonin concentrations can be mathematically obtained for any placement of fibers and their varicosities by setting all parameters to biologically realistic values. Our goal is to provide a foundational description of this system, while developing an exhaustive toolset, through which, in the longer term, spatiotemporal serotonin pattern formation can be found when additional chemical species that locally react with serotonin are introduced. Our 2-D CRD system can be solved numerically and with low computational cost.

This article is organized as follows. Section \ref{sec:model} introduces a two-dimensional CRD system modeling a serotonergic fiber network, first in dimensional form with biologically realistic parameters and then in non-dimensional form. Section \ref{sec:main} refines this system by taking the limit of infinite permeabilities of the varicosity neighborhoods and averaging over variable firing periods, which allows us to compute long-term steady-state serotonin concentrations, when averaged over each firing period, both locally and across space, as well as to estimate the peak level of serotonin within the neighborhood of a firing varicosity. In Section \ref{sec:arrangements}, we illustrate these results using systems of five diffusively coupled fibers, each with one to three varicosities, where subsets of varicosities fire at biologically relevant frequencies. Section \ref{sec:comparison_linear} compares Michaelis–Menten uptake kinetics \cite{handy2021revising} with linear uptake. Finally, we end with a discussion of our developed framework, results, and biological implications in Section \ref{sec:discussion}.

\section{Model setup}
\label{sec:model}
We start by considering a thin slice of neural tissue that contains fibers that have discrete varicosity sites that both release and continuously take up serotonin (Figure~\ref{fig:serotoninsystem}D). Specifically, we consider $k\in\{1, ..., M\}$ 1-D fibers placed in $\mathbb{R}^2$ with a minimum distance of $L_\cF$ between them. Along these fibers, $N$ discrete varicosity sites are placed at points $\vX_j\in \mathbb{R}^2$, $j\in\{1, ..., N\}$, with a minimum distance of $L$ between them. We let $\cF_k$ denote the set of varicosity sites that are located along fiber $k$. For example, if $\vX_1$, $\vX_6$, and $\vX_{11}$ are the varicosity sites along the $k^\text{th}$ fiber, then $\cF_k = \{\vX_1, \vX_6, \vX_{11}\}$. 

Around each discrete varicosity site, we put disks of radius $R_0$ in $\mathbb{R}^2$, denoted as $\Omega_j$. These disk regions, when assumed to be well-mixed, allow for mathematically incorporating various kinds of local behavior close to the varicosity sites. As a result, we can describe the total amount of serotonin at time $\tau$ inside disk $\Omega_j$, denoted as $\mu_j(\tau)$, by the following ordinary differential equation (ODE)
\begin{eqnarray} \label{eq:mujODE}
    && \frac{d\mu_j}{d\tau} = F_\text{uptake}(\mu_j) + F_{\text{release},j}(\tau) + C(\mu_j,\mathcal{U}),
\end{eqnarray}
where the first two terms describe the uptake and release dynamics occurring within each $\Omega_j$, respectively, and the last term, $C(\mu_j, \mathcal{U})$, connects the disk $\Omega_j$ to the diffusing serotonin concentration $\mathcal{U}(\tau, \vX)$ in the surrounding extracellular space. 

Serotonin uptake is mediated by membrane-bound serotonin transporters that require the formation of a transporter-ligand complex and therefore can saturate. We follow the result of~\cite{handy2021revising,Yousofvand2025} and model this with a Michaelis-Menten reaction term 
\begin{eqnarray*}
    F_\text{uptake}(\mu_j) := -\frac{\tilde{V}\mu_j}{\tilde{K} + \mu_j},
\end{eqnarray*}
where $\tilde{V}$ is the maximal uptake rate and $\tilde{K}$ is the Michaelis-Menten half-saturation constant. When fiber $\cF_k$ fires at time $T_l^{(k)}$ with $l$ indexing the firing event, serotonin is quickly released from membrane-fused synaptic vesicles, which we model as
\begin{eqnarray*}
    F_{\text{release},j}(\tau) := r\mu_c \sum_l k_R\delta_\theta\left(k_R\tau - k_R \,T_l^{(k(j))}\right),
\end{eqnarray*}
with $k_R$ denoting a representative dimensional reaction rate, $r\mu_c$ expressing the total amount of serotonin released into the system with each pulse, $\mu_c$ representing the typical magnitude of serotonin in a varicosity neighborhood, and $k(j)$ selecting the fiber number to varicosity $j$. To incorporate the quick release of serotonin into the compartment, we choose $\delta_\theta$ to be a Gaussian probability density function with a small standard deviation $\theta$, which scales the duration of a single firing event,
\begin{equation*}
    \delta_\theta(k_R\tau) := \frac{1}{\sqrt{2\pi\theta^2}} \exp\left(-\frac{k_R^2}{2\theta^2} \tau^2\right),
\end{equation*}
though, importantly, the results of this work hold as long as the function $\delta_\theta$ is sufficiently smooth. Here, we choose the Gaussian profile for simplicity and analytical convenience.

Assuming linear exchange at the boundary of the disk with \emph{influx rate times disk area per disk boundary length} $\beta_1$ and \emph{efflux rate per disk boundary length} $\beta_2$, the last term on the right-hand side of \eqref{eq:mujODE} becomes
\begin{eqnarray*}
    C(\mu_j, \mathcal{U}) = \int_{\partial \Omega_j} (\beta_1 \cU(\tau, \vX) - \beta_2 \mu_j(\tau)) \; dS_\vX,
\end{eqnarray*}
where we integrate over all exchange across the boundary of $\Omega_j$ denoted by $\partial \Omega_j$.

To complete our mathematical formulation, we need to specify the dynamics in the extracellular space between the disks. Here, we assume that serotonin simply diffuses and degrades, which will diffusively couple the local serotonin behavior in all $\Omega_j$, $j \in \{1, ..., N\}$. The natural partial differential equation (PDE) and boundary condition for this is
\begin{eqnarray*}
    && \partial_\tau \mathcal{U} = D_B \Delta \mathcal{U} - k_B \mathcal{U}, \quad \vX \in \bR^2\backslash \bigcup_{j=1}^N \Omega_j, \\
    && D_B \partial_{n_\vX} \mathcal{U} = \beta_1 \mathcal{U} - \beta_2 \mu_j, \quad \vX \in \partial \Omega_j,
\end{eqnarray*}
where $D_B$ denotes the diffusion coefficient, $k_B > 0$ the degradation rate, and $n_\vX$ is the outward normal vector to $\Omega_j$.

In total, the corresponding compartmental-reaction diffusion system described above, assuming there is no extracellular serotonin concentration initially, is 
\begin{subequations} \label{eq:dimPDEsystem}
\begin{eqnarray}
    \partial_\tau \mathcal{U} &=& D_B \Delta \mathcal{U} - k_B \mathcal{U}, \quad \tau > 0, \quad  \vX \in \bR^2\backslash \bigcup_{j=1}^N \Omega_j, \\
    \mathcal{U}(0, \vX) &=& 0, \quad \mathcal{U}(\tau, \vX) \to 0 \text{ as } |\vX| \to \infty, \\ [4pt]
    D_B \partial_{n_\vX} \mathcal{U} &=& \beta_1 \mathcal{U} - \beta_2 \mu_j, \qquad \vX \in \partial \Omega_j, \\ [4pt]
    \frac{d\mu_j}{d\tau} &=& F(\mu_j,\tau) +\int_{\partial \Omega_j} (\beta_1 \mathcal{U}(\tau, \vX) - \beta_2 \mu_j(\tau)) \; dS_\vX, \quad j\in\{1, ..., N\},
\end{eqnarray}
with intra-disk (intra-compartmental) kinetics 
\begin{equation}\label{eq:intrakin_simple_dim}
    F(\mu_j,\tau) := -\frac{\tilde{V}\mu_j}{\tilde{K} + \mu_j}+r \mu_c \sum_l k_R\,\delta_\theta(k_R\tau - k_R\, T_l^{(k(j))}) \quad \text{for} \quad \vX_j \in {\cF_k}, k \in \{1, ..., M\},
\end{equation}
and release function
\begin{equation} \label{eq:delta_release}
    \delta_\theta(k_R \tau) = \frac{1}{\sqrt{2\pi\theta^2}} \exp\left({-\frac{k_R^2}{2\theta^2}\tau^2}\right).
\end{equation}
\end{subequations}
\subsection{Biological parameters} \label{subsec:bparams}
The biological parameter values are based on available mouse and human data. These values continue to be refined with novel experimental techniques, and some estimates are more accurate than others (as discussed below). However, the available quantitative information is sufficient for capturing the key aspects of the microscale dynamics of the brain serotonergic system. We note that the properties of the brain serotonergic system have been highly conserved in vertebrate evolution ~\cite{Janusonis2025, moroz_neural_2021}.

The density of serotonergic fibers (including varicosities) varies in different brain regions~\cite{Awasthi2021}. The diameter of a typical varicosity is around 1 µm and varicosities are often spaced by 2-\SI{4}{\micro\meter} along the fiber~\cite{haiman2023qualitative,mays2023experimental}. The spacing of serotonergic fibers in a 2-D projection of a \SI{40}{\micro\meter}-section can be approximated by a regular grid with lines separated by around \SI{5}{\micro\meter}. This estimate is based on the mouse amygdala, a high fiber-density region~\cite{mays2023experimental}, and produces an estimate of around $2.5$-$5.0\cdot10^6$ varicosities in \SI{1}{\milli\meter}$^3$ of neural tissue. It is in good agreement with an early estimate obtained in the rat cerebral cortex (around $6\cdot10^6$ varicosities/\SI{}{\milli\meter}$^3$)~\cite{jacobs1992structure} and with a recent estimate in the subthalamic nucleus of the cynomolgus macaque \textit{(Macaca fascicularis)~\cite{delmas2025serotonin}.} 

In the model, we assume that each varicosity is a point surrounded by a well-mixed neighborhood of the radius $R_0=$ \SI{0.3}{\micro\meter} (Figure~\ref{fig:serotoninsystem}D, right), and that any two varicosity neighborhoods (on the same or different fibers) are spaced at least by $L=$ \SI{4}{\micro\meter}. Since the model is two-dimensional, the third dimension (``tissue thickness") is considered to be negligible with respect to the relevant dynamics. In physical terms, we assume it to be approximately \SI{1}{\micro\meter}. In 2-D-projections at this thickness, serotonergic fibers become much more separated compared to 2-D-projections of \SI{40}{\micro\meter}-thick sections. Based on estimates in digital \SI{1}{\micro\meter}-thick substacks of confocal images of the mouse amygdala~\cite{mays2023experimental}, we set this value to $L_\cF =$ \SI{16}{\micro\meter}.

Serotonergic neurons often fire at around 1-3 Hz and can maintain a relatively regular interspike interval~\cite{jacobs1992structure}. However, they can produce considerably higher firing rates that have been associated with alertness, sensory stimuli, and positive reinforcement~\cite{Zhang2025b}.

When a serotonergic neuron fires, a total amount of $r \cdot \mu_c$ of serotonin is pulse-released into a small ball (neighborhood) surrounding each varicosity. Inside a varicosity, serotonin molecules are packed into many synaptic vesicles, and each synaptic vesicle may contain around 5,000-60,000 molecules~\cite{Redpath2022}. However, only 0 to 2 vesicles are released by a typical presynaptic bouton (or a varicosity) in response to an action potential~\cite{Maschi2020,Zhang2025b}. In the model, we assume that each varicosity releases around 15,000 serotonin molecules, which corresponds to $250 \cdot 10^{-22}$ moles. The timescale of this release is in the millisecond range~\cite{Sudhof2013,Chanaday2018}; in the model, we set it at $4\theta/k_R =0.0064$ s. 

The reuptake of serotonin from the ECS into varicosities is driven by SERT and depends on two parameters, the Michaelis-Menten constant and the maximal reuptake velocity. The Michaelis-Menten constant is the serotonin amount at which SERT molecules are half-saturated. Because it reflects the inherent affinity of serotonin to SERT, it can be accurately estimated in blood platelets that are easily accessible and express the same SERT. Based on previous studies, this values falls in the range of 0.6-0.9 \SI{}{\micro M}~\cite{anderson2002serotonin,brown2023palmitoylation}. In our model, these molecules are located only in the small neighborhoods around the varicosities, and must have the same units as $\mu_j$. As a result, we must account for the volume of these neighborhoods when setting the value of this half-saturation constant. Hence, we choose $\tilde{K}\sim 7.5\cdot 10^{-7}$ M$\cdot Vol \approx 0.85\cdot 10^{-22}$ mol, where the volume of the 3-D sphere of our 2-D projected disk about a varicosity is $Vol= \frac{4}{3}\pi(0.3^3\cdot 10^{-15})$ L. 

The second reuptake parameter, the maximal reuptake velocity $\tilde{V}$, is more difficult to estimate from platelet data. This quantity depends on the number of available SERT molecules in the varicosity membrane and can in theory vary not only across tissues but also across developmental ages, brain regions, and exposures to pharmacological agents. However, serotonergic varicosities and blood platelets do share remarkable similarities~\cite{Janusonis2014}. In human platelets, $\tilde{V}$ has been estimated to be around 1.8 nmol/min per $10^9$ platelets, or around $3 \cdot 10^{-20}$ mol/s in one platelet~\cite{franke2010serotonin}. Another indirect estimate can be obtained from a study that has expressed the human SERT in an in vitro system~\cite{brown2023palmitoylation}. This study has reported a $\tilde{V}$ value of 27 pmol/(min$\cdot$mg [protein]). Assuming that brain tissue has 10\% protein (by weight)~\cite{Naber1979} and that the specific weight of brain tissue is around 1 g/cm$^3$~\cite{Barber1970}, this value corresponds to $4.5\cdot 10^{-14}$ mol/s in 1 mm$^3$ of brain tissue. Further assuming that this volume contains around $6 \cdot 10^6$ serotonergic varicosities, this estimate becomes around $0.75 \cdot 10^{-20}$ mol/s in one varicosity. The two values, obtained indirectly in very different systems, are of the same order and quite similar. In this work, we set $\tilde{V} = 2 \cdot 10^{-20}$ mol/s.

Lastly, we estimate two parameters governing serotonin dynamics in the ECS, the serotonin diffusivity ($D_B$) and the serotonin degradation rate ($k_B$). We base this former value on published estimates that have factored in the ECS tortuosity \cite{rice1985diffusion, hentall2006spatial}, thus yielding an ``effective'' diffusivity, and set it to $0.3 \cdot 10^{-5}$ cm$^2$/s ($300$~\SI{}{\micro\meter}$^2$/s). 

In addition to SERT, serotonin can be removed from the ECS by non-serotonergic cells (e.g., astrocytes)~\cite{Zhao2023} and/or broken down outside serotonergic neurons. This removal rate is difficult to estimate experimentally. Since this is the only ``free" parameter in the model, we set it to a small value ($k_B$ = 0.75 s$^{-1}$) relative to the varicosity reuptake rate. With this value, our model yields theoretical concentrations that closely approximate experimentally observed serotonin concentrations in the ECS (0.2-70 nM)~\cite{Zhao2023}. We note that these degraded serotonin molecules are continuously replaced by synthesized new molecules. Brain serotonin synthesis rates have been estimated by experimental studies~\cite{Nishizawa1997}, but they also include serotonin that is broken down inside serotonergic neurons by monoamine oxidases~\cite{Zhao2023}, as part of the normal recycling process. We effectively capture this undepletable serotonin pool by assuming a constant serotonin amount released per action potential.           

All parameter choices are summarized in Table \ref{tab:parameters}.

\begin{table}[h!]
\footnotesize
\caption{The default parameter values for the dimensional system (\ref{eq:dimPDEsystem}).} \label{tab:parameters}
\centering
\begin{tabular}{||l l l l||} 
 \hline
 Parameter & Description & Value & Unit \\ [0.5ex] 
 \hline\hline
 $R_0$ & radius of varicosity neighborhood & 0.3 & \text{µm}  \\ 
 \hline
 $L$ & Minimum varicosity neighborhood separation & 4 & \text{µm} \\
 \hline
  $L_\cF$ & Parallel fiber separation & 16 & \text{µm} \\
 \hline
 $D_B$ & Effective serotonin diffusivity (incl. tortuosity) & 300 & $\frac{\text{µm}^2}{\text{s}}$ \\
 \hline
 $k_B$ & Extracellular serotonin degradation rate & 0.75 & $\text{s}^{-1}$ \\
\hline
 $\tilde{K}$ & Michaelis-Menten half-saturation constant & 0.85$\cdot 10^{-22}$ & mol \\
 \hline
 $\tilde{V}$ & Michaelis-Menten maximal velocity parameter & 2$\cdot 10^{-20}$ & $\frac{\text{mol}}{\text{s}}$ \\ 
 \hline
 $r\cdot \mu_c$ & Amount of locally released serotonin  & $2.5 \cdot 10^{-20}$ & mol \\
 \hline
 $\mu_c$ & Typical magnitude of serotonin & $10^{-22}$ & mol \\
 \hline
 $4\theta/k_R$ & Duration of release & 6.4 $\cdot 10^{-3}$ &  s \\ 
 \hline
$k_R$ & Representative dimensional reaction rate & 18.75 & s$^{-1}$  \\ 
 \hline
\end{tabular}
\end{table}

\subsection{Non-dimensionalization} \label{subsec:nondim}
In this work, we will use strong localized perturbation theory (see Section \ref{sec:main}), which requires a small spatial parameter that highlights the scale separation between the disk size of the neighborhoods and the distance between these neighborhoods. As a result, we perform a thorough non-dimensionalization of our system to clearly identify this scaling, along with other potential key parameters of our system. If we denote the dimension operator by $\left[\cdot\right]$, the dimensions of the variables and parameters in \eqref{eq:dimPDEsystem} are
\begin{eqnarray*}
    && [\tau] = \text{s}, \quad [\vX] = \text{µm}, \quad [\cU] = \frac{\text{mol}}{\text{µm}^2}, \quad [D_B] = \frac{\text{µm}^2}{\text{s}}, \quad [k_R] = \frac{1}{\text{s}} = [k_B], \\
    && \left[\beta_1\right] = \frac{\text{µm}}{\text{s}}, \quad [\beta_2] = \frac{1}{\text{µm}\cdot \text{s}}, \quad [\mu_j] = \text{mol} = [\mu_c], \quad [\tilde{V}] = \frac{\text{mol}}{\text{s}}, \quad [\tilde{K}] = \text{mol}.
\end{eqnarray*}

Following the non-dimensionalization as in \cite{pelzward2025synchronized}, we scale space by our minimum inter-compartment distance ($L$), time by an intra-compartmental reaction kinetics ($k_R$, selected to be 18.75 s$^{-1}$), and the typical magnitude of serotonin in a varicosity neighborhood ($\mu_c$), yielding our non-dimensional parameters
\begin{eqnarray*}
    &&t := k_R \,\tau, \quad \vx := \frac{\vX}{L}, \quad U := L^2\frac{\cU}{\mu_c}, \quad u_j := \frac{\mu_j}{\mu_c}, \quad \varepsilon := \frac{R_0}{L} = 0.075, \\
    && D := \frac{D_B}{k_R L^2} = 1, \quad d_{1} := \varepsilon \frac{\beta_{1}}{k_RL} = \cO(1), \quad d_{2} := \varepsilon \frac{\beta_{2}L}{k_R} = \cO(1), \quad \sigma := \frac{k_B}{k_R},
\end{eqnarray*}
where we highlight that $\varepsilon = R_0/L=0.075$ is the small spatial parameter critical to our theory. The non-dimensional intra-compartmental reaction kinetic function is then
\begin{equation}
\begin{split}
    f(t, u_j) := -\frac{V u_j}{K + u_j} + r \sum_l \delta_\theta(t - k_R \, T_l^{(k(j))}), \quad \text{with} \quad V = \frac{\tilde{V}}{k_R\mu_c} \quad \text{and} \quad K= \frac{\tilde{K}}{\mu_c}.
\end{split}
\end{equation}
It also follows that the ratio of the permeabilities is $c = d_{2}/d_{1} = L^2 \cdot \beta_{2}/\beta_{1}$, where $d_{2}/d_{1} = \mathcal{O}(1)$. In the simple case that the influx rate and the efflux rate at each point on the surface of the full 3-dimensional sphere with volume $Vol$ corresponding to the 2-dimensional projected disk $\Omega_j$ are constant and equal, say $\beta > 0$ with unit $[\beta]=\frac{1}{\text{s}}$, the ratio of the rates $\beta_{1}$ and $\beta_{2}$ is
\begin{equation*}
    \beta_{2} /\beta_{1} = \frac{\beta}{2\pi R_0} \, / \,\frac{\beta\pi R_0^2}{2\pi R_0} = \frac{1}{\pi R_0^2}.
\end{equation*}
To see this, note that if the concentration inside the projected disk $\mu_j/(\pi R_0^2)$ is equal to the surrounding concentration $\mathcal{U}$, $\beta_1 \mathcal{U} - \beta_2 \mu_j = 0$ close to the disk boundary if the influx and efflux rates are the same. In this way, the non-dimensional ratio $c$ is
\begin{equation*}
    c = \frac{L^2}{\pi R_0^2} \approx 56.59.
\end{equation*}

In total, the final non-dimensional system reads
\begin{subequations} \label{eq:PDEsystem}
\begin{eqnarray}
    \partial_t U &=& D \Delta U - \sigma U, \quad t > 0, \quad  \vx \in \bR^2\backslash \bigcup_{j=1}^N B_\varepsilon(\vx_j), \\
    U(0, \vx) &=& 0, \quad U(t, \vx) \to 0 \text{ as } |\vx| \to \infty, \\ [4pt]
    \varepsilon D \partial_{n_\vx} U &=& d_1 U - d_2 u_j, \qquad \vx \in \partial B_\varepsilon(\vx_j), \\ [4pt]
    \frac{d u_j}{dt} &=& f(t, u_j) +\frac{1}{\varepsilon}\int_{\partial B_\varepsilon(\vx_j)} (d_1 U(t, \vx) - d_2 u_j(t)) \; dS_\vx, \quad j\in\{1, ..., N\}, \nonumber
\end{eqnarray}
with intra-disk (intra-compartmental) kinetics 
\begin{equation}\label{eq:intrakin_simple}
    f(t, u_j) := -\frac{V u_j}{K + u_j} + r \sum_l \delta_\theta(t - k_R T_l^{(k)}) \quad \text{for} \quad \vx_j \in \cF_k, k \in \{1, ..., M\},
\end{equation}
and release function
\begin{equation} \label{eq:delta_release_nondim}
    \delta_\theta(t) = \frac{1}{\sqrt{2\pi\theta^2}} \exp\left(-\frac{t^2}{2\theta^2}\right).
\end{equation}
\end{subequations}
All non-dimensional parameter values for this model are summarized in Table~\ref{tab:parameters_cont}.

\begin{table}[h!]
\footnotesize
\caption{The non-dimensional parameter values for the system \eqref{eq:PDEsystem}.} \label{tab:parameters_cont}
\centering
\begin{tabular}{||l l l||} 
 \hline
 Parameter & Description & Value \\ [0.5ex] 
 \hline\hline
 $D$ & Non-dim. serotonin diffusivity & 1 \\
 \hline
 $\sigma$ & Non-dim. degradation rate & 0.04 \\
 \hline
 $\varepsilon$ & Non-dim. local disk radius & 0.075 \\
 \hline
 $K$ & Non-dim. Michaelis-Menten half-saturation constant & 0.85 \\
 \hline
 $V$ & Non-dim. Michaelis-Menten maximal velocity parameter & 10.67\\
 \hline
 $c$ & Concentration conversion constant & 56.59 \\
 \hline
 $r$ & Non-dim. released serotonin & 250 \\
 \hline
 $\theta$ & Standard deviation of Gaussian release function & 0.03 \\  [1ex] 
\hline
\end{tabular}
\end{table}

\section{Main results} \label{sec:main}
Using strong localized perturbation theory, an asymptotically equivalent integro-ODE system for all intra-ball amounts $u_j$ has been derived in \cite{pelzward2025synchronized}. In this way, precise coupling functions $B_j(t)$ have been derived that capture the diffusive coupling of the projected 2-D disks. This integro-ODE system is extremely helpful for computing an accurate numerical solution, because numerically solving the above PDE-ODE system \eqref{eq:PDEsystem} is difficult and time-consuming on the order of 16 h for only two diffusively coupled disks\footnote{Run with FlexPDE 7 © by PDE Solutions Inc on a virtual machine (VM) server with Ubuntu 16.04.7 LTS, VM CPU Intel\textcopyright Core\texttrademark (Haswell) 2299.996 MHz with 4 cores, 16GB VM RAM, and real CPU of the host server Intel\textregistered Xeon\textregistered CPU E5-2698 v3 @ 2.30 GHz}. The corresponding integro-ODE system including all information about the extracellular serotonin dynamics $U(t, \vx)$ is summarized in the following proposition, adapted from \cite{pelzward2025synchronized} to the case of a single intra-compartmental and diffusing species.

\begin{proposition}[Adapted from Pelz and Ward, 2025, \cite{pelzward2025synchronized}] \label{prop:2ode} For $\varepsilon\to 0$, and with the initial condition $U(0, \vx)=0$, the solution $U(t, \vx)$ and $u_j(t)$ of \eqref{eq:PDEsystem}, for $j\in\lbrace{1,\ldots,N\rbrace}$, is approximated for $t\gg \cO (\varepsilon^2)$
  by 
  \begin{subequations}\label{2:reduced}
  \begin{align}
    \frac{du_j}{dt} &= f(t, u_j) + B_j(t) \,, 
    \label{2:reduced_1} \\
    \int_{0}^{t} B_j^{\prime}(s) E_1(\sigma(t-s))\, ds &= \eta B_{j}(t) + \gamma u_{j}(t) \nonumber \\
          & \qquad + \sum_{\stackrel{k=1}{k\neq j}}^{N} \int_{0}^{t} \frac{B_{k}(s) e^{-\sigma(t-s)}}{t-s} e^{-|\vx_j-\vx_k|^2/(4D(t-s))} \, ds\,,\label{2:reduced_2}
  \end{align}
  for $j\in\lbrace{1,\ldots,N\rbrace}$, where $E_1(z) = \int_z^\infty \frac{exp(-s)}{s} \; ds$ is the exponential integral. Further, in this integro-differential
  system $\eta$ and $\gamma$ are defined by
  \begin{equation}\label{2:reduced_3}
    \eta := 2 \left( \frac{1}{\psi} + \frac{D}{d_{1}} +
      \log\left(2 \sqrt{\frac{D}{\sigma} }\right) - \gamma_e \right)=
    -\log(\varepsilon^2\kappa_{0}\sigma) \,, \qquad
    \gamma := \frac{4\pi D d_{2}}{d_{1}} \,,
  \end{equation}
\end{subequations}
where $\kappa_{0} := \frac{1}{4D} e^{2\left(\gamma_e - {D/d_{1}}\right)}$ with Euler's constant $\gamma_e$ and $\psi:=-1/\log(\varepsilon)$.  In terms of
$B_{j}(t)$, the approximate solution in the bulk region is
\begin{equation}\label{2:reduced_bulk}
  U(t, \vx) \sim -\frac{1}{4\pi D} \sum_{j=1}^{N}\int_{0}^{t}
  \frac{B_{j}(s) e^{-\sigma(t-s)}}{t-s}
  e^{-|\vx-\vx_j|^2/(4D(t-s))} \, ds \,,
\end{equation}
while in the vicinity of the $j^{\mbox{th}}$ cell we have for
$\rho=\varepsilon^{-1}|\vx-\vx_j|={\mathcal O}(1)$ that
\begin{equation}\label{2:reduced_local}
  U \sim\frac{B_{j}(t)}{2\pi D} \log\rho +
  \frac{B_{j}(t)}{2\pi d_{1}} + \frac{d_{2}}{d_{1}} u(t) \,.
\end{equation}
\end{proposition}

Additionally, in \cite{pelzward2025synchronized} a marching scheme is developed which allows for numerically solving the system \eqref{2:reduced} efficiently for up to 100 varicosities ($\sim 1$ h runtime for the full-time solution for 50 sites). Once solved, a video of extracellular serotonin $U(t, \vx)$ can be created efficiently for any grid in $\vx$ with any given fineness, to get an intuition for the global serotonin concentration dynamics. \\

\subsection{Instant exchange between compartments} Since the real system does not consist of any bounded ball centered around each varicosity, we eliminate the ball boundaries to obtain instant exchange of local serotonin $u_j$ with the extracellular serotonin concentration $U$. It can be seen in \eqref{2:reduced} that the coupled limit of the boundary permeabilities $d_1 \to \infty$, $d_2 \to \infty$, $d_2/d_1 = c$, for some ratio $c$, is attainable. This ratio $c$ depends on the area of the projected disks, as $u_j$ corresponds to total serotonin \emph{amount} in the projected disk $B_\varepsilon(\vx_j)$ and $U(t, \vx)$ corresponds to the serotonin \emph{concentration} at $\vx$ at time $t$. This heuristic statement is more rigorously derived in the following Corollary.

\begin{corollary}[System with realistic exchange between local and global serotonin]\label{cor:2ode}
The coupled PDE-ODE system in the limit of infinite permeabilities, $d_1 \to \infty$ and $d_2 \to \infty$ such that $\frac{d_2}{d_1} = c$ is fixed,
\begin{subequations}\label{eq:PDEsystem_infty}
\begin{eqnarray}
    \partial_t U = D\Delta U - \sigma U, \quad t > 0, \quad \vx\in\bR^d\backslash\bigcup_{j=1}^N B_\varepsilon (\vx_j), \\
    \int_{\partial B_\varepsilon(\vx_j)} U(t, \vx)\; dS_\vx = 2\pi\varepsilon c\, u_j(t), \quad U(0, \vx) = 0, \quad U(t, \vx) \to 0 \text{ as } |\vx| \to \infty, &&\\
    \frac{du_j}{dt} = f(t, u_j) + \int_{\partial B_\varepsilon(\vx_j)} D\partial_n U(t, \vx)\; dS_\vx, \quad j\in\{1, ..., N\}  ,\label{eq:PDEsystem_infty_localkin}
\end{eqnarray}
\end{subequations}
for $\varepsilon \to 0$ and with initial condition $U(0, \vx) = 0$ is approximated for $t \gg \cO(\varepsilon^2)$ by 
\begin{subequations} \label{eq:reduced_infty}
\begin{align} 
    \frac{du_j}{dt} &= f(t, u_j) + B_j(t) \,, 
    \label{2:reduced_1_serot} \\
    \int_{0}^{t} B_j^{\prime}(s) E_1(\sigma(t-s))\, ds &= \eta_\infty B_{j}(t) + \gamma_\infty u_{j}(t) \nonumber \\
          & \;\;\;\; + \sum_{\stackrel{k=1}{k\neq j}}^{N} \int_{0}^{t} \frac{B_{k}(s) e^{-\sigma(t-s)}}{t-s} e^{-|\vx_j-\vx_k|^2/(4D(t-s))} \, ds\,,\label{2:reduced_2_serot}
  \end{align}
  for $j\in\lbrace{1,\ldots,N\rbrace}$. In this integro-differential
  system, $\eta_\infty$ and $\gamma_\infty$ are defined by
  \begin{equation}\label{2:reduced_3_serot}
    \eta_\infty := 2 \left( \frac{1}{\psi} +
      \log\left(2 \sqrt{\frac{D}{\sigma} }\right) - \gamma_e \right)=
    -\log(\varepsilon^2\kappa^\infty_{0}\sigma) \,, \qquad
    \gamma_\infty := 4\pi D c \,,
  \end{equation}
\end{subequations}
where $\kappa^\infty_{0} := \frac{1}{4D} e^{2\gamma_e}$ with Euler's constant $\gamma_e$ and $\psi:=-1/\log(\varepsilon)$.  In terms of
$B_{j}(t)$, the approximate solution in the bulk region is
\begin{equation}\label{2:reduced_bulk_serot}
  U(t, \vx) \sim -\frac{1}{4\pi D} \sum_{j=1}^{N}\int_{0}^{t}
  \frac{B_{j}(s) e^{-\sigma(t-s)}}{t-s}
  e^{-|\vx-\vx_j|^2/(4D(t-s))} \, ds \,,
\end{equation}
while in the vicinity of the $j^{\mbox{th}}$ cell we have for
$\rho=\varepsilon^{-1}|\vx-\vx_j|={\mathcal O}(1)$ that
\begin{equation}\label{2:reduced_local_serot}
  U \sim\frac{B_{j}(t)}{2\pi D} \log\rho +
  c\, u_j(t) \,.
\end{equation}
\end{corollary}
\begin{proof}
    We extend the proof of Proposition~\ref{prop:2ode} provided in \cite{pelzward2025synchronized}, using the modified PDE-ODE system \eqref{eq:PDEsystem_infty} to account for this fast exchange between compartments and their surrounding space. For $t=\cO(1)$, we zoom into the local region of the $j^\text{th}$ ball via the coordinate change $\vy_j := \varepsilon^{-1}(\vx-\vx_j)$ with $\rho_j := |\vy|$, so that the local function $V_j(t, \vy_j) = U(t, \vx_j + \varepsilon \vy_j)$ satisfies the leading-order quasi-steady problem
    \begin{equation}\label{eq:local_sys_infty}
        \Delta_{\vy_j} V_j = 0 \;\; \text{for} \;\; \rho_j>1, \quad \int_{\partial B_1(\vy_j)} V_j(t, \vy_j)\; dS_{\vy_j} = 2\pi c u_j(t) \;\; \text{on} \;\; \rho_j=1.
    \end{equation}
    In terms of some $B_j(t)$ to be found, the radially symmetric solution to \eqref{eq:local_sys_infty} is
    \begin{equation}
        V_j(t, \vy_j) = \frac{B_j(t)}{2\pi D} \log|\vy_j| + c \, u_j(t),
    \end{equation}
    and, plugged into the local kinetics \eqref{eq:PDEsystem_infty_localkin}, we obtain
    \begin{equation}
        \frac{du_j}{dt} = f(t, u_j) + B_j(t), \quad j\in\{1, ..., N\}.
    \end{equation}
    In order to find the coupling function $B_j(t)$, we need to asymptotically match to the global solution $U$ in the original coordinates satisfying
    \begin{equation}
        \begin{split}
        &\partial_t U = D\Delta U - \sigma U,  t>0, \quad \vx\in\bR^2\backslash\{\vx_1, ..., \vx_N\}, \quad U(0, \vx) = 0, \\
        &U \sim \frac{B_j}{2\pi D} \log|\vx-\vx_j| + \frac{B_j}{2\pi D \psi} + c \, u_j \quad \text{as} \quad \vx \to \vx_j, \quad j\in\{1, ..., N\},
        \end{split}
    \end{equation}
    with $U(t, \vx)\to 0$ as $|\vx|\to \infty$ for all $t>0$. Using Laplace transformation and a Green function, the construction of the global solution can be done exactly as in \cite{pelzward2025synchronized} where the matching to the local behavior now yields
    \begin{align}
        \begin{split}
        \quad \frac{B_j(t)}{2\pi D \psi} + c u_j(t) =& \frac{B_j(t)}{2\pi D} \left(\gamma_e - \log\left(2\sqrt{\frac{D}{\sigma}}\right)\right) + \frac{1}{4\pi D} \int_0^t B_j'(s) E_1(\sigma(t-s))\; ds \\
        &- \sum_{\stackrel{k=1}{k\neq j}}^{N} \int_0^t \frac{B_k(s) e^{-\sigma(t-s)}}{4\pi D (t-s)} e^{-|\vx_j-\vx_k|^2/(4D(t-s))}\; ds \quad,
        \end{split}
    \end{align}
    instead of (3.35) in \cite{pelzward2025synchronized}, where $\gamma_e$ is Euler's constant coming from the small-argument asymptotics of the Green function. Combining constants into $\eta_\infty$, $\gamma_\infty$, and $\kappa_0^\infty$ completes the proof.
\end{proof}

The mathematical limit derived above establishes an effective description of instantaneous exchange between local varicosity serotonin levels and the extracellular field. A natural next question is which statistical distributions of fiber firing times can generate nontrivial spatial patterns of serotonin and how the parameters of these distributions influence the outcome. In particular, the strength of spatial coupling among fibers determines whether diffusion reinforces or smooths local heterogeneities in serotonin levels. For the reaction kinetics given in \eqref{eq:intrakin_simple}, the interplay between continuous local uptake and diffusive spread after each firing event produces transient fluctuations near active sites during low-frequency firing, while promoting more stable extracellular reservoirs between and around the sites at higher frequencies. This observation motivates the subsequent analysis of period-averaged dynamics, where the long-term behavior can be characterized more systematically.

\subsection{Period-averaged dynamics}
For periodic or random firing with respectively a fixed period or a fixed distribution, we can attempt to recover the local period-averaged or mean serotonin level at each varicosity the system converges to. This will allow us to compare such amounts between neighboring varicosities. In the following we focus on deterministic periodic firing, leaving stochastic firings for a future study (though see Appendix~\ref{poissonSec} and Section~\ref{poissonSM} for a cursory look). This repeated serotonergic firing rules out the existence of steady-states with their usual definition, which is a major difference of this current study compared to previous work~\cite{pelz2023symmetry,pelzward2025synchronized}. However, by applying the correct time-dependent operator on our system, we can make the existence of steady-states possible, which can be solved for analytically. This would provide additional checks of our numerical results and intuition. 

Specifically, we will consider the period-averaging operator $\sfE_t^\freq$ defined by  $\sfE^\freq_t[\, g(\cdot)\,] := \freq \int_t^{t+1/\freq} g(\tilde{t}) \; d\tilde{t}$ for an integrable function $g$ on $[t, t+1/\mathsf{f}]$, where $\freq$ denotes the firing frequency of the fiber. Applying this operator to the non-dimensional integro-differential system \eqref{eq:reduced_infty} for the coupled local serotonin amounts with Michaelis-Menten kinetics gives the period-averaged system
\begin{subequations} \label{eq:reduced_infty_averaged}
\begin{align}
    \frac{d\sfE^\freq_t[u_j]}{dt} &= -V\sfE^\freq_t\left[\frac{u_j}{K + u_j}\right] + r\freq(\mathbbm{1}_\freq)_j + \sfE^\freq_t[B_j] \,, 
    \label{2:reduced_1_ave} \\
    \sfE^\freq_t\left[\int_{0}^{\,\cdot} B_j^{\prime}(s) E_1(\sigma(\cdot-s))\, ds \right] &= \eta_\infty \sfE^\freq_t[B_{j}] + \gamma_\infty \sfE^\freq_t[u_{j}] \nonumber \\
          & \qquad + \sum_{\stackrel{k=1}{k\neq j}}^{N} \sfE^\freq_t\left[\int_{0}^{\,\cdot} \frac{B_{k}(s) e^{-\sigma(\cdot-s)}}{\cdot-s} e^{-|\vx_j-\vx_k|^2/(4D(\cdot-s))} \, ds\right]\,,\label{2:reduced_2_ave}
\end{align}
\end{subequations}
for $j\in\lbrace{1,\ldots,N\rbrace}$ and where $(\mathbbm{1}_\freq)_j$ is the $j$th entry in the indicator vector $\mathbbm{1}_\freq \in \{0,1\}^N$, which is 1 if varicosity $j$ is periodically firing and 0 otherwise. It is now central to derive whether $\sfE^\freq_t[u_j]$ settles to a finite steady-state and what bounds on the steady-state can be obtained. First we present a small helpful result that states that we can move the period-averaging operator into the convolution if we are interested in the long-term dynamics, before deriving a corresponding lower bound on $\sfE^\freq_t[u_j]$, and eventually showing that a unique steady-state exists.

\begin{lemma} \label{lem:t-ave_convolution_commute}
    For a bounded oscillatory function $g(t)$ with oscillation frequency $\freq$ and a steadily decaying function $h(t)$ with $\lim_{t\to\infty} h(t) = 0$, it holds that
    \begin{equation*}
        \sfE^\freq_t\left[\int_{0}^{\cdot} g(s) h(\cdot-s) \, ds\right] = \int_0^t \sfE^\freq_{t-s}[g] \, h(s) \, ds + R_t,
    \end{equation*}
    where the remainder $R_t$ satisfies $\lim_{t\to\infty} R_t = 0$ and is bounded by $ |R(t)| \leq \frac{1}{\freq}|h(t)| \,\sfE_0^\freq[|g|]$ for large enough $t$.
\end{lemma}
\begin{proof}
    We compute
        \begin{align*}
            \sfE^\freq_t\biggl[\int_{0}^{\,\cdot} &g(s) h(\cdot-s)ds \biggr] = \freq\int_t^{t+1/\freq} \int_{0}^{\tilde{t}} g(\tilde{t}-s) h(s) \, ds \, d\tilde{t} \\   
            &= \freq\int_0^{1/\freq} \int_{0}^{\vartheta + t} g(t-s + \vartheta) h(s) \, ds \, d\vartheta\\
            &= \freq\int_0^{1/\freq} \left(\int_{0}^{t} g(t-s + \vartheta) h(s)\, ds + \int_{t}^{\vartheta + t} g(t-s + \vartheta) h(s)\, ds\right)\, d\vartheta \\   
            &= \int_0^t \left(\freq\int_0^{1/\freq} g(t-s+\vartheta) \, d\vartheta\right) \, h(s) \, ds + \mathsf{f}\int_0^{1/\mathsf{f}} \int_0^\vartheta g(\vartheta-\theta) h(t+\theta) \; d\theta \; d\vartheta,   
        \end{align*}    
        where we used the substitution $\vartheta := \tilde{t} - t$ to arrive at the second line, split the inner integral to achieve the third line, and grouped/rearranged terms and made the substitution $\theta := s - t$ to arrive at the final equality. The first integral in the final equality is just $\int_0^t \sfE_{t-s}^\freq[g]\, h(s)\, ds$. The latter integral can be, with the substitution $\varphi:= \vartheta - \theta$, written as
        \begin{equation*}
            R(t) := \freq \int_0^{1/\freq} \int_0^{\vartheta} g(\varphi) \, h(t+\vartheta - \varphi)\, d\varphi \;d\vartheta.
        \end{equation*}
        This vanishes as $t\to \infty$ due to $h(t) \to 0$ and is, for large enough $t$, bounded by
        \begin{displaymath}
            |R(t)| \leq |h(t)| \int_0^{1/\freq} \freq \int_0^{\vartheta} g(\varphi) \, d\varphi \;d\vartheta \leq |h(t)| \int_0^{1/\mathsf{f}} |g(\varphi)| \, d\varphi \leq  \frac{1}{\freq}|h(t)| \,\sfE_0^\freq[|g|].
        \end{displaymath}
\end{proof}

\begin{proposition}[Period-averaged steady-state] \label{prop:period-ave}
    In the case that a subset of varicosities is firing with a single and steady frequency $\freq$, not necessarily all in-phase, while the remaining (could be none) varicosities do not fire, a finite non-negative steady-state of the period-averaged system \eqref{eq:reduced_infty_averaged} exists, is unique, and, as $t\to\infty$, the system \eqref{eq:reduced_infty_averaged} can be rewritten as
    \begin{subequations} \label{eq:reduced_infty_averaged_longt}
\begin{align}
    \frac{d\sfE^\freq_t[u_j]}{dt} &= -V\sfE^\freq_t\left[\frac{u_j}{K + u_j}\right] + r\freq(\mathbbm{1}_\freq)_j + \sfE^\freq_t[B_j] \,, 
    \label{2:reduced_1_longt} \\
    \int_{0}^{t} (B_j(t-s+1/\freq) - B_j(t-s)) \;E_1(\sigma s)\, ds &= \eta_\infty \sfE^\freq_t[B_{j}] + \gamma_\infty \sfE^\freq_t[u_{j}] \nonumber \\
          & \qquad + \sum_{\stackrel{k=1}{k\neq j}}^{N} \int_{0}^{t} \sfE^\freq_{t-s}[B_{k}] \;\frac{e^{-\sigma s}}{s} e^{-|\vx_j-\vx_k|^2/(4Ds)} \, ds\,,\label{2:reduced_2_longt}
\end{align}
\end{subequations}
for $j \in \{1, ..., N\}$.
    
Since $-\frac{u_j}{K+u_j}$ is strictly convex with $K>0$, a lower bound for $\sfE^\freq_t[u_j]$ can be obtained by solving for the steady-state of \eqref{eq:reduced_infty_averaged_longt} with $-V\sfE^\freq_t\left[\frac{u_j}{K + u_j}\right]$ replaced by $-V\frac{\sfE^\freq_t[u_j]}{K + \sfE^\freq_t[u_j]}$, where we rename $\sfE^\freq_t[u_j]$ to $\sfE^\freq_t[v_j]$ in the latter, i.e., $\sfE^\freq_t[v_j] < \sfE^\freq_t[u_j]$. The values of $\sfE^\freq_t[v_j]$ can readily be obtained fast by numerically finding the asymptotic steady-state of \eqref{eq:PDEsystem} with $f(t, u_j) = -\frac{Vu_j}{K+u_j} + r\freq$ via the corresponding nonlinear algebraic system
    \begin{equation*}
        -\frac{V\sfE^\freq_\infty[v_j]}{K + \sfE^\freq_\infty[v_j]} + r\freq(\mathbbm{1}_\freq)_j -\gamma_\infty{\bf{e}}_j^T\mathrm{M}^{-1}\sfE^\freq_\infty[{\bf{v}}] = 0,
    \end{equation*}
    where the vector $\mathbf{v}$ has entries $v_j$ and $\mathrm{M}$ is the matrix with entries
    \begin{equation*}
    \mathrm{M}_{jk} =
    \begin{cases}
        \eta_\infty, & j = k, \\[4pt]
        2K_0\!\left(\sqrt{\frac{\sigma}{D}}\,|\mathbf{x}_j - \mathbf{x}_k|\right), & j \ne k,
    \end{cases}
    \end{equation*}
    for $j,k \in \{1, ..., N\}$ and with $K_0$ denoting the modified Bessel function of the second kind. In this way, the non-negative steady-state $\sfE^\freq_\infty[{\bf{u}}]$ is bounded by 
    \begin{equation*}
        \sfE^\freq_\infty[{\bf{v}}] \leq \sfE^\freq_\infty[{\bf{u}}] \leq \frac{r}{\gamma_\infty} \freq \mathrm{M} \mathbbm{1}_\freq,
    \end{equation*}
    where ${\bf{u}}$ is the vector with entries $u_j$.
\end{proposition}
\begin{proof}
    We begin by applying Lemma \eqref{lem:t-ave_convolution_commute} to both convolution integrals, which yields \eqref{eq:reduced_infty_averaged_longt}.
    
    The steady-state $\sfE^\freq_\infty[v_j]$ of
    \begin{align*}
    \frac{d\sfE^\freq_t[v_j]}{dt} &= -\frac{V\sfE^\freq_t[v_j]}{K + \sfE^\freq_t[v_j]} + r\freq(\mathbbm{1}_\freq)_j + \sfE^\freq_t[\tilde{B}_j] \,, 
     \\
    \int_{0}^{t} \sfE^\freq_{t-s}[\tilde{B}_j'] \;E_1(\sigma s)\, ds &= \eta_\infty \sfE^\freq_t[\tilde{B}_{j}] + \gamma_\infty \sfE^\freq_t[u_{j}] \nonumber \\
          & \qquad + \sum_{\stackrel{k=1}{k\neq j}}^{N} \int_{0}^{t} \sfE^\freq_{t-s}[\tilde{B}_{k}] \;\frac{e^{-\sigma s}}{s} e^{-|\vx_j-\vx_k|^2/(4Ds)} \, ds,
    \end{align*}
    if it exists, solves the nonlinear algebraic system
    \begin{align*}
        -\frac{V\sfE^\freq_\infty[v_j]}{K + \sfE^\freq_\infty[v_j]} + r\freq(\mathbbm{1}_\freq)_j + \sfE^\freq_\infty[\tilde{B}_j] = 0, \\
        \eta_\infty \sfE^\freq_\infty[\tilde{B}_j] +\gamma_\infty \sfE^\freq_\infty[v_j] + 2\sum_{\overset{k=1}{k\neq j}}^N \sfE^\freq_\infty[\tilde{B}_k] K_0\left(\sqrt{\frac{\sigma}{D}} |\vx_j-\vx_k|\right) = 0,
    \end{align*}
    as derived in \cite{pelzward2025synchronized}. The coupling function $\sfE^\freq_\infty[\tilde{B}_j]$ can be eliminated noting that
    \begin{equation*}
        \mathrm{M}\sfE^\freq_\infty[{\bf{\tilde{B}}}] = -\gamma_\infty\sfE^\freq_\infty[{\bf{v}}],
    \end{equation*}
    where $\mathrm{M}$ is the matrix with entries
    \begin{equation*}
        \mathrm{M}_{jj} = \eta_\infty, \quad \mathrm{M}_{jk} = 2K_0\left(\sqrt{\frac{\sigma}{D}} |\vx_j-\vx_k|\right),
    \end{equation*}
    $j \neq k \in \{1, ..., N\}$, and ${\bf{\tilde{B}}}$ and ${\bf{v}}$ are vectors with entries $\tilde{B}_j$ and $v_j$. Importantly, $\mathrm{M}$ is invertible for a generic configuration of varicosities and parameter values. To see this, note that since $K_0$ is real-analytic on $(0,\infty)$ and the varicosities are at distinct locations, $F(\vx)=\det\mathrm{M}(\vx)$ is a real-analytic function on the connected open set of $\mathcal{C}_N = \{\vx \in \mathbb{R}^{2N}: \vx_i \ne \vx_j \text{ for all } i\ne j\}$. Further, we see that $F(\vx)$ is not identically zero, since for sufficiently spaced out points, $\mathrm{M}$ is strictly diagonally dominant, because $K_0$ decays exponentially and $\eta_\infty\ne 0$. The latter point is true since for generic diffusivities $D$, cell radii $\varepsilon$, and degradations $\sigma$, it holds that $\eta_\infty = -\log(\varepsilon^2\sigma \frac{1}{4D}e^{2\gamma_e})$ does not equal 0. Since the zero set of a nontrivial real-analytic function on a connected open set is closed, nowhere dense, and has measure zero, the result that $\mathrm{M}$ is generically invertible follows. This result gives
    \begin{equation*}
        -\frac{V\sfE^\freq_\infty[v_j]}{K + \sfE^\freq_\infty[v_j]} + r\freq(\mathbbm{1}_\freq)_j -\gamma_\infty{\bf{e}}_j^T\mathrm{M}^{-1}\sfE^\freq_\infty[{\bf{v}}] = 0,
    \end{equation*}
    with standard unit normal vector ${\bf{e}}_j$.
    Jensen's inequality yields 
    \begin{equation*}
        -V\, \sfE^\freq_t\left[ \frac{u_j}{K+u_j}\right] > -V \frac{\sfE^\freq_t[u_j]}{K + \sfE^\freq_t[u_j]},
    \end{equation*}
    which implies that it needs to hold that $\sfE^\freq_t[u_j] > \sfE^\freq_t[v_j]$ for all $t>0$ for the equation to be satisfied. Thus, $\sfE^\freq_t[v_j]$ is truly a lower bound on $\sfE^\freq_t[u_j]$ due to the convexity of $u_j/(K+u_j)$ in $u_j$.
    
    It only remains to show that a unique non-negative steady-state $\sfE^\freq_\infty[u_j]$ of \eqref{eq:reduced_infty_averaged_longt} exists, because, from $0 < \sfE^\freq_\infty[v_j] < \sfE^\freq_\infty[u_j]$ and the non-oscillatory nature of $\sfE^\freq_t[v_j]$, the existence of $\sfE^\freq_\infty[v_j]$ would follow from the existence of $\sfE^\freq_\infty[u_j]$, $j\in\{1, ..., N\}$. Non-negativity of $\mathbf{u}(t)$ follows from system~\eqref{eq:PDEsystem}, where one can see that the initial data satisfies $\mathbf{u}(0) = \mathbf{0}$, the release term $r\sum_l \delta_\theta(\cdot)\ge0$ can increase these amounts, while the reuptake term $-Vu_j/(K+u_j)$ vanishes at $u_j=0$, preventing the solution from crossing zero. Second, similarly as for $\sfE^\freq_\infty[v_j]$, a steady-state $\sfE^\freq_\infty[{\bf{u}}]$ solves
    \begin{align*}
        -V\sfE^\freq_\infty\left[\frac{u_j}{K + u_j}\right] + r\freq(\mathbbm{1}_\freq)_j -\gamma_\infty{\bf{e}}_j^T\mathrm{M}^{-1}\sfE^\freq_\infty[{\bf{u}}] = 0,
    \end{align*}
    which can be rewritten as
    \begin{align*}
        \sfE^\freq_\infty[\mathbf{u}] = \frac{\mathrm{M}}{\gamma_\infty}\left(r\freq\mathbbm{1}_\freq-V\sfE_\infty^\freq\left[\mathbf{u}\oslash\left(K\mathbbm{1}+\mathbf{u}\right)\right]\right),
    \end{align*}
    where $\oslash$ denotes element-wise division. Since the right-hand side is monotone in $\mathbf{u}$, a unique steady-state of $\sfE^\freq_\infty[\mathbf{u}]$ exists.
\end{proof}

\subsection{Spike max-min estimates for large times} \label{subsec:ampl_est}
For steady firing of serotonin with a certain frequency $\freq$, it is of interest what the firing-amplitude of serotonin amount is around each varicosity, as well as the nearby non-firing varicosities. We will give an approximate answer to this question in the case of long-term firing (i.e., after a few tens of seconds). In experimental and clinical neuroscience, knowing the approximate long-term period-averages $\sfE^\freq_\infty[v_j]$ along with approximate long-term spiking maximal values provides a key characterization of the long-term dynamics. The following lemma provides an estimate for long-term spiking maximal serotonin values of the firing varicosities.

\begin{lemma}[Estimate for long-term spiking maxima] \label{lem:longterm-spikemax}
    For the integro-ODE system \eqref{eq:reduced_infty} to the PDE-ODE system \eqref{eq:PDEsystem_infty} for the serotonin dynamics with a subset of varicosities steady firing at a frequency $\freq$, an estimate for the spiking maxima in the varicosity neighborhoods is provided as $t\to \infty$ by
    \begin{subequations}
        \begin{equation}
        \begin{aligned}
        \vv_{\max} =  \sfE_{\infty}^\freq[\vv] &+ \left(I - \frac{4\Delta T}{5} \gamma_\infty \tilde{M}^{-1}\right)^{-1}\left(-V\sfE^\freq_\infty[\vv]\oslash\left(K\mathbbm{1} + \sfE^\freq_\infty[\vv]\right) \Delta T \right.\\
        &\left.+ (2\Phi(2) -1 + \freq \Delta T)\,r\mathbbm{1}_\freq + \sfE_\infty^\freq[\tilde{\vB}]\Delta T\right),
        \end{aligned}
        \end{equation}
    for $j\in\{1, ..., N\}$, where the lower bound vector $\vv$ together with vector $\tilde{\vB}$ have respective entries $v_j$ and $\tilde{B}_j$ from Proposition \ref{prop:period-ave} and $\oslash$ denoting element-wise division. The duration $\Delta T = 4\theta$ is the approximate duration of a serotonin spike $\delta_\theta$, $\Phi$ is the cumulative distribution function of the standard Gaussian distribution, and $\gamma_\infty$ and $\eta_\infty$ are given in \eqref{2:reduced_3}. The matrix $\tilde{M}$ has entries 
    \begin{equation}
        \tilde{M}_{jj} = E_1(\sigma\Delta T) + \frac{1 - e^{-\sigma\Delta T}}{\sigma\Delta T} - \eta_\infty \quad \text{and} \quad \tilde{M}_{jk} = -\frac{1}{2} \mathcal{I}_{jk}(\Delta T), \quad j\neq k,
    \end{equation}
    where 
    \begin{equation}
        \mathcal{I}_{jk}(\Delta T) :=\left(1+\frac{|\vx_j-\vx_k|^2}{4D/\sigma} \right) E_1\left(\frac{|\vx_j-\vx_k|^2}{4D\Delta T}\right) - e^{-|\vx_j-\vx_k|^2/(4D\Delta T)}\sigma\Delta T.
    \end{equation}
    \end{subequations}
\end{lemma}
For the derivation of the estimate, see Appendix \ref{spikemaxDeriv}.
\begin{remark}
    It remains to address the tightness of the estimate of Lemma \ref{lem:longterm-spikemax} to the true long-term spike maxima. This will be examined in Section~\ref{sec:oneRow}.
\end{remark}
From this estimate for the spike maxima, an estimate for spike minima is easily obtained. Since the period-averaged steady-state $\sfE_\infty^\freq[\vv]$ is known from Proposition \ref{prop:period-ave}, a simple, but not tight, estimate for the long-term spike minima is given by $\max\{2\sfE_\infty^\freq[\vv] - \vv_{\max}, 0\}$.

\section{Numerical experiments} 
\label{sec:arrangements}

We now apply our theory from the previous section to example fiber arrangements. We start by considering a single line of varicosities spread across different fibers (Section \ref{sec:oneRow}), and then build upon this intuition by layering rows of varicosities (Section \ref{sec:multipleRows}). We then compare our nonlinear Michaelis-Menten condition with a linear counterpart (Section~\ref{sec:comparison_linear}). Throughout this section, we solve the integro-ODE system \eqref{eq:reduced_infty} with an accurate and computationally efficient marching scheme similar to the one developed in \cite{pelzward2025synchronized} (see Appendix \ref{app:numerics}). And unless otherwise noted, we also refer to and plot dimensional quantities of all variables, namely $\mathcal{U} = U\cdot\mu_c/L^2$,  $\mu_j=u_j\cdot\mu_c$ and $\nu_j = v_j\cdot\mu_c$.

\subsection{Single row of varicosities}\label{sec:oneRow}
Despite a simplification, examining one row of varicosities spread out across multiple fibers can be justified as follows. First, consider a set of vertical fibers (Figure \ref{fig:example_fibers}, left). Viewing the deviation from a straight line as a perturbation, we can straighten them out (Figure \ref{fig:example_fibers}, middle) and then focus only on a 1-D section through the 2-D setup with equidistantly placed varicosities (Figure \ref{fig:example_fibers}, right), viewing contributions of the other varicosities as minor and non-equidistantly separated varicosities again as perturbations. This simplified setup allows us to analyze key contributions of firing neighboring fibers to the serotonin distributions around varicosities of other fibers all orthogonal to the section. 

\begin{figure}[H]
    \centering
    \includegraphics[width=0.85\textwidth]{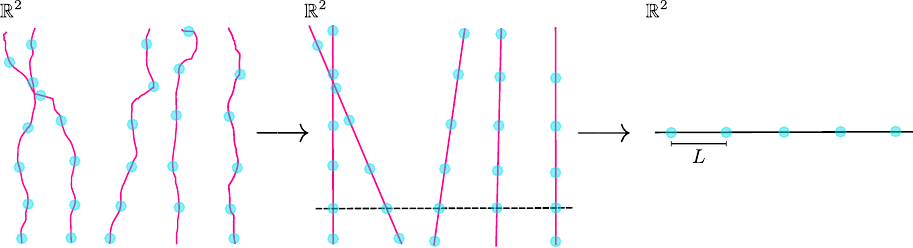}
    \caption{A system in $\bR^2$ with five serotonergic fibers (axons) in magenta and varicosity neighborhoods in cyan (left) gets straightened (middle), viewing deviations from a straight line as perturbations. We then focus only on a 1-D section (black dashed in middle; black on right) orthogonal to the fibers with equidistantly placed varicosities, in order to solely analyze the contribution of neighboring firing varicosities of different fibers on all other varicosities along the section. The minimal distance between varicosities is defined to be $L$, while the radius of their neighborhoods is $R_0$. See Table~\ref{tab:parameters} for default parameter values used throughout this work.}
    \label{fig:example_fibers}
\end{figure}

With this geometry in mind, we start by only letting the fibers furthest left ($j = 1$) and furthest right ($j = 5$) of the five fibers fire at 2 Hz, which amounts in our reduced system with the 1-D section in $\bR^2$ to the left and right varicosity firing. We will begin by considering three snapshots of the serotonin concentration in time: 1) during the first firing event, 2) in the middle of the time interval until the second firing event, and 3) just before the start of the second firing event). First, we find that when the serotonin is initially released, the concentration at and around varicosities on the firing fibers ($j = 1$ and $5$) transiently increase (Figure \ref{fig:snapshots_lineboundaryfiring2Hz}, top-left panel). Serotonin then diffuses away and gets removed locally at all varicosities via Michaelis-Menten reuptake (i.e., SERT), so much as that, by the next time points, the system shows little serotonin remaining in the 1-D slice (Figure \ref{fig:snapshots_lineboundaryfiring2Hz}, top-middle/right panels). 
\begin{figure}[!hbtp] 
    \centering
    \includegraphics[width=0.9\linewidth]{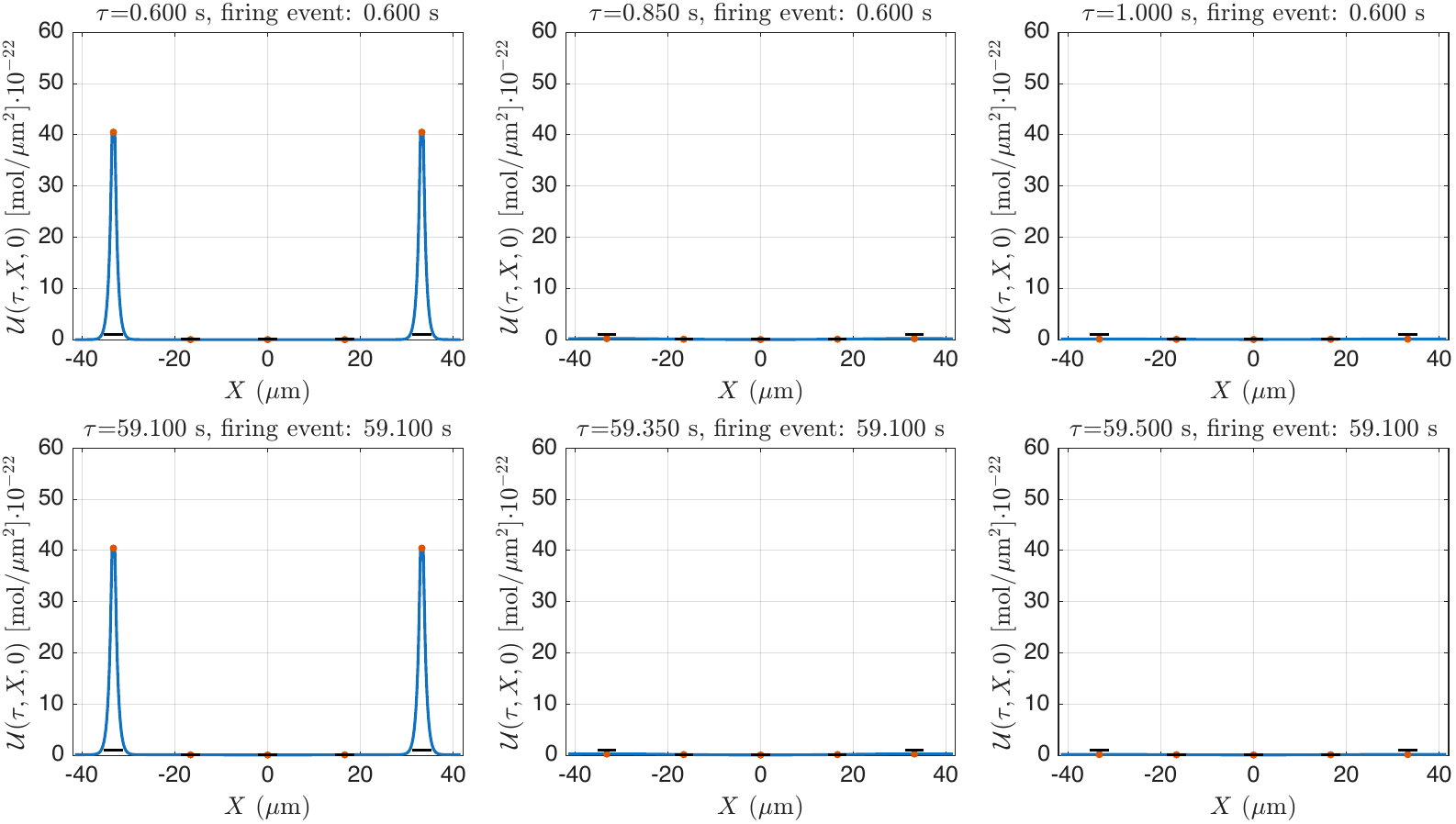}
    \caption{Top row: three snapshots of the dynamic serotonin distribution at the first firing event with \SI{2}{\Hz}) (left; firing midpoint $+0 \cdot 1/\freq$), in the middle of the time interval until the second firing event (middle; firing midpoint $+1/2 \cdot 1/\freq$), and right at the start of the next firing event (right; firing midpoint $+8/10 \cdot 1/\freq$). Bottom row: similar snapshots but after a firing event at which the period-averaged dynamics are close to their steady-state. The lower bound of the period-averaged steady-state, $\sfE^\freq_\infty[\nu_j]\cdot c/L^2$, computed as in the proof of Proposition \ref{prop:period-ave} is added into the snapshot with horizontal black bars. All plotted solutions were computed rewriting the 2-D heat kernel in \eqref{2:reduced_bulk_serot} for the extracellular concentration as a sum of exponentials and using the marching scheme outlined in Appendix \ref{app:numerics}. The plotted solution is the global solution matched with the local solution (red dots). All other parameters are set to their default values (see Table~\ref{tab:parameters})}
    \label{fig:snapshots_lineboundaryfiring2Hz}
\end{figure}

Considering similar snapshots during a firing event at a later point in time (i.e., after the system has been allowed to equilibrate), we find similar dynamics as before (Figure \ref{fig:snapshots_lineboundaryfiring2Hz}, bottom row). This indicates that, under this geometry and firing rates, the serotonin concentration is able to return close to baseline, without spreading significantly to nearby varicosities. This is further demonstrated by considering the lower-bound steady-states, converted to be in the appropriate units, $\sfE_\infty^\freq[\nu_j]\cdot c/L^2$ for $j\in\{1, ..., 5\}$, of Proposition \ref{prop:period-ave}, which remain close (though slightly above) to baseline for $j\in\{2, 3, 4\}$ (Figure \ref{fig:snapshots_lineboundaryfiring2Hz}, small horizontal black bars).

This investigation can be repeated for the same geometry but under different firing frequencies. We will specifically consider the fibers firing at 16 Hz firing and 64 Hz (Figures \ref{fig:snapshots_lineboundaryfiring16Hz} and \ref{fig:snapshots_lineboundaryfiring64Hz}, respectively). In these cases, we see that the firing rate is high enough such that the serotonin reaches the middle three varicosities, leading to a significant serotonin `reservoir' that lasts between the firing events.   
\begin{figure}[!hbtp] 
    \centering
    \includegraphics[width=0.85\linewidth]{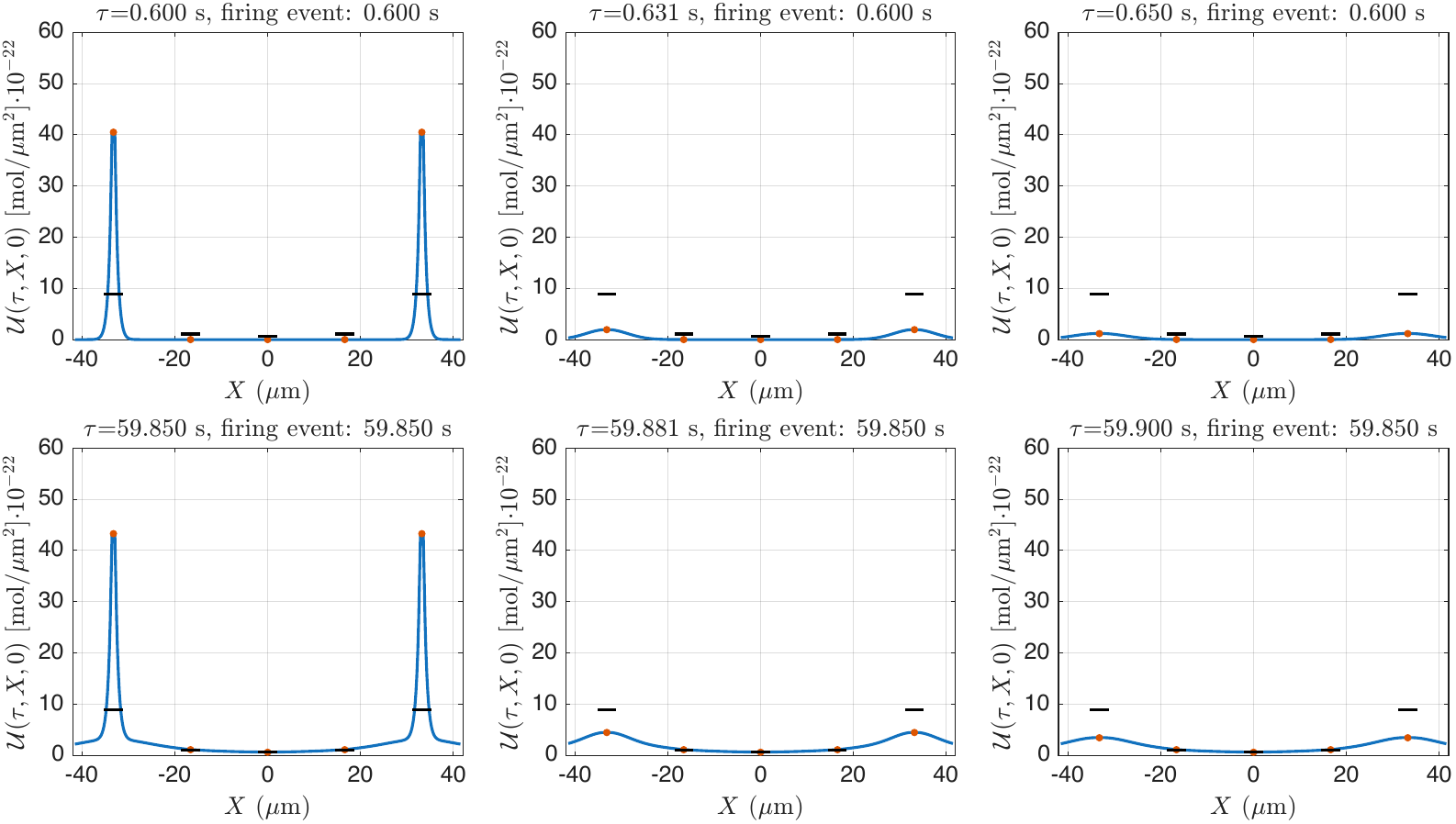}
    \caption{Same caption as Figure~\ref{fig:snapshots_lineboundaryfiring2Hz}, but with firing at \SI{16}{\Hz}. An emerging serotonin reservoir can be observed between the firing varicosities.}
    \label{fig:snapshots_lineboundaryfiring16Hz}
\end{figure}
\begin{figure}[!hbtp] 
    \centering
    \includegraphics[width=0.85\linewidth]{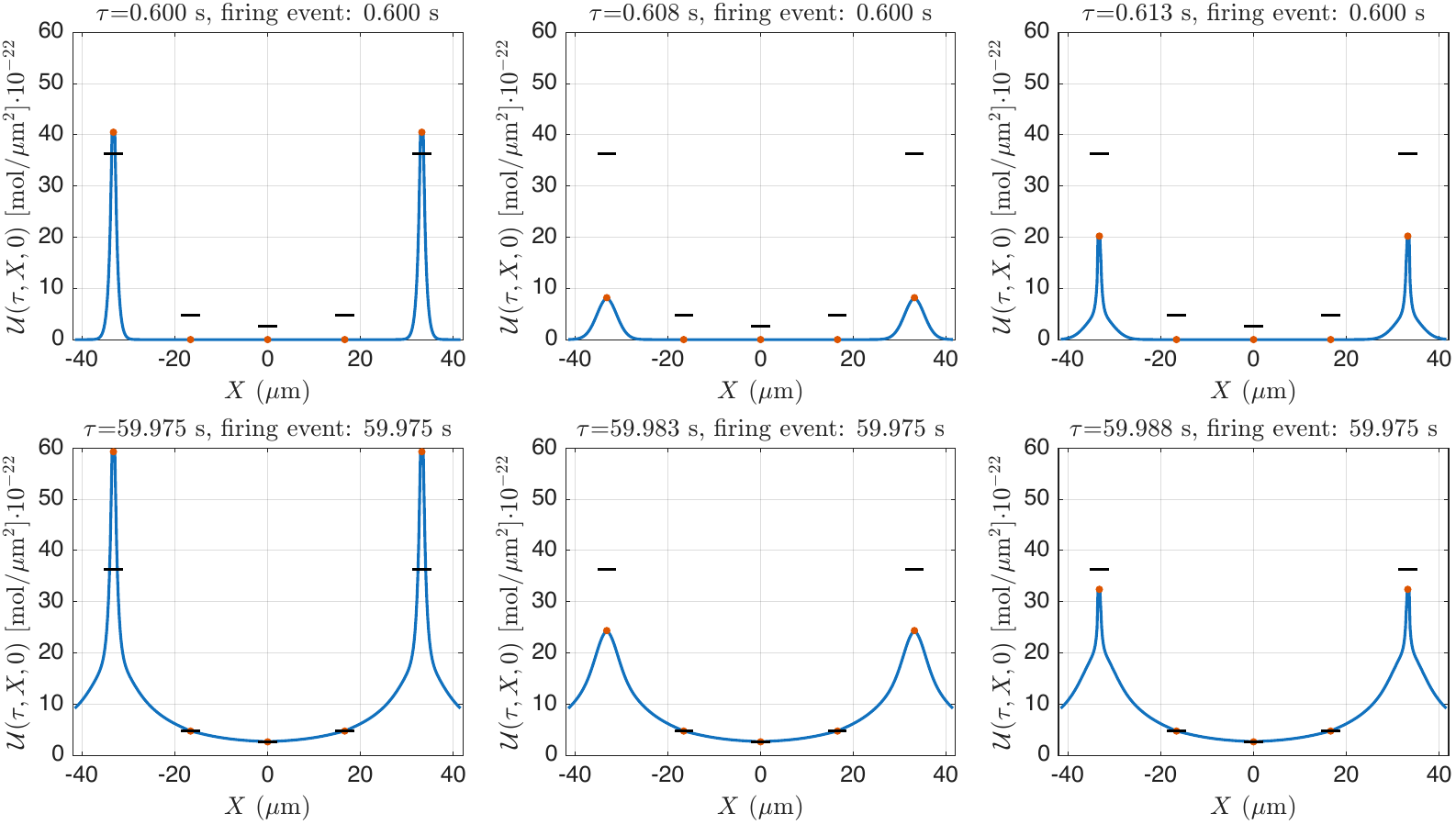}
    \caption{Same caption as Figure~\ref{fig:snapshots_lineboundaryfiring2Hz}, but with firing at \SI{64}{\Hz}. The serotonin reservoir between the firing varicosities reaches higher levels.}
    \label{fig:snapshots_lineboundaryfiring64Hz}
\end{figure}

Continuing with the same geometry and firing configuration (i.e., varicosities at $j=1$ and $j=5$ fire at 2 Hz), we now examine the dynamics across time \textit{within} each of the five varicosity neighborhoods. These correspond to the height of the magenta dots in Figure~\ref{fig:snapshots_lineboundaryfiring2Hz}, scaled by $L^2/c$ to convert concentrations to total amounts (Figure~\ref{fig:2Hz}). In both the initial transient (left column) and several periods later, once the system has settled into its periodic steady-state (right column), $\mu_j(\tau)$ is dominated by sharp spikes at the firing sites $j=1$ and $j=5$, with smaller, delayed responses at the non-firing sites hidden beneath them on the shared vertical scale. The symmetry of the arrangement implies that the dynamics at $j = 1$ and $j = 2$ should exactly match those of $j = 5$ and $j = 4$, respectively, and our numerical simulations confirm this.
\begin{figure}[!hbtp]
    \centering
	\includegraphics[width=.8\linewidth]{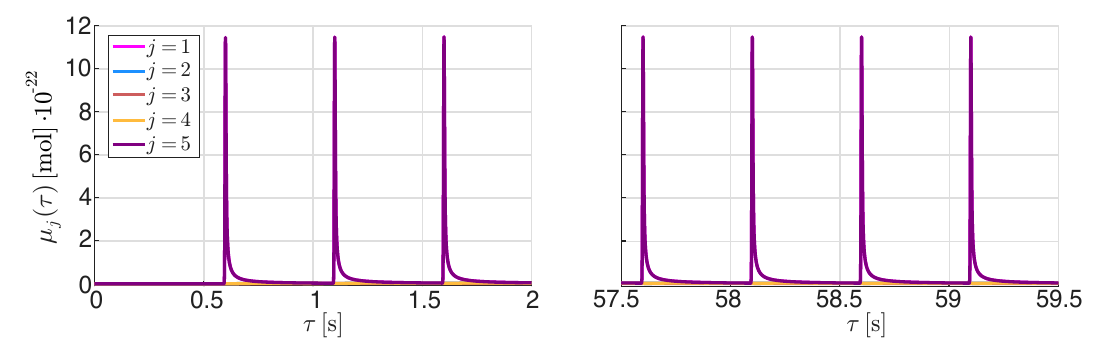}
    \caption{The dynamics of total concentrations in the five varicosity neighborhoods along the 1-D section of Figure \ref{fig:example_fibers} when the varicosity sites at $j = 1$ and $j = 5$ fire in-phase at \SI{2}{\Hz} for the initial transient (left) and several periods later (right). Due to the symmetry of the arrangement, the curves for $j=1$ and $j=2$ are precisely on the ones for $j=5$ and $j=4$, respectively. All other parameters are set to their default values (see Table~\ref{tab:parameters}).}
    \label{fig:2Hz}
\end{figure}

Figure~\ref{fig:2HzCompare} provides a closer look at these dynamics on a per-site basis, alongside two reference quantities from Proposition~\ref{prop:period-ave}: the lower bound $\sfE^\freq_\infty[\nu_j]$ (dashed lines) and the numerically estimated period-average $\sfE^\freq_\tau[\mu_j]$ (dash-dotted lines). The latter is obtained by first solving numerically for $\mu_j$ and then approximating the integral $\int_\tau^{\tau+1/\mathrm{f}}\mu_j(s) \; ds$ by quadrature to yield $\sfE^\freq_\tau[\mu_j]$. At this resolution, the subtle increases at non-firing sites and the time lag in $\mu_j(\tau)$ across $j$ become clearly visible. Because the system starts from rest, the numerically computed period-average $\sfE^\freq_\tau[\mu_j]$ initially lies below $\sfE^\freq_\infty[\nu_j]$ (left column). It then ramps up toward the bound as the system approaches its periodic steady-state. By the time window shown in the right column, the two quantities more closely agree. Note that the y-axis is zoomed in for clarity, and the actual gap between the lower bound and the numerical estimate is quite small.
\begin{figure}[!hbtp]
    \centering
    \includegraphics[width=0.9\linewidth]{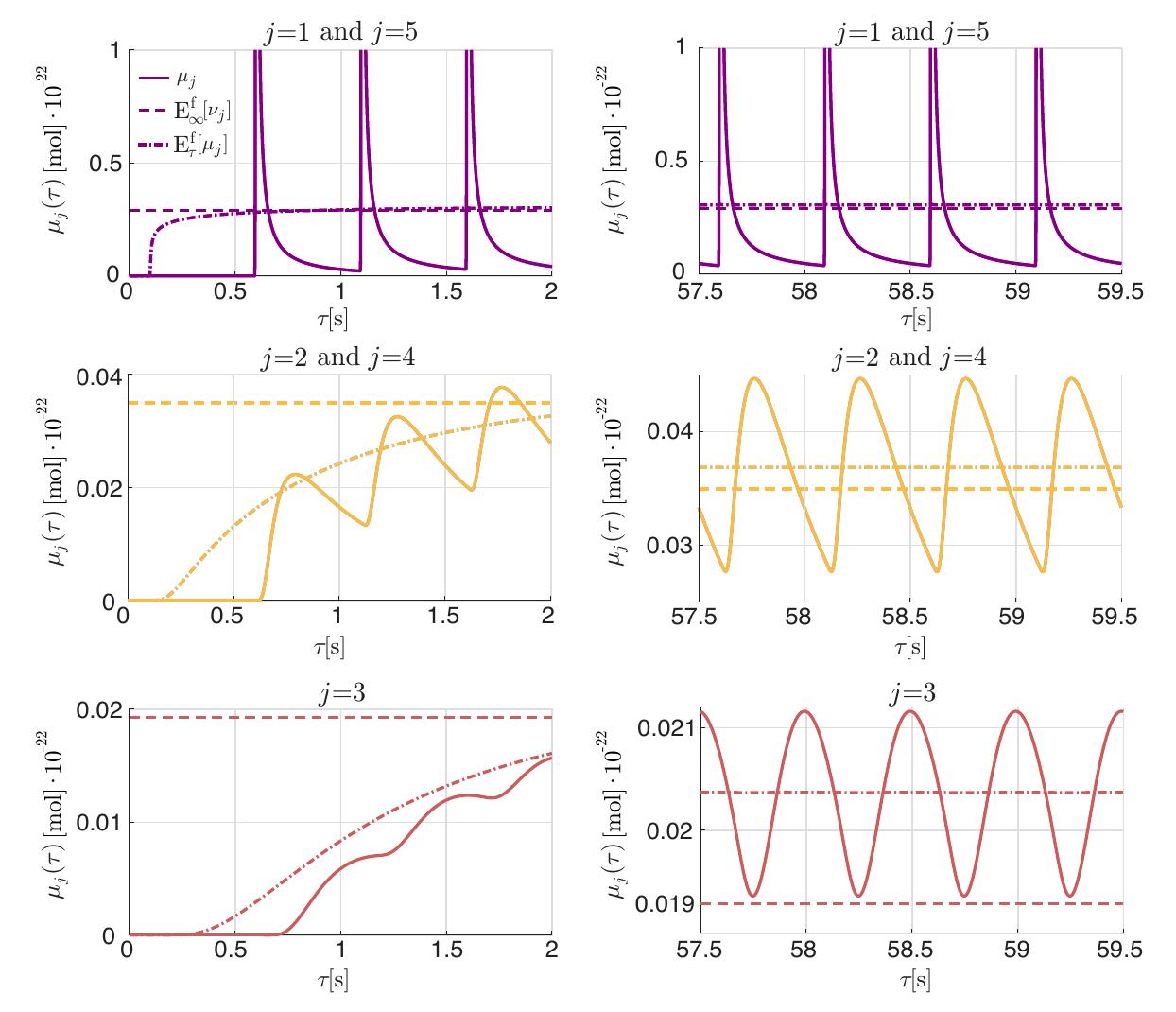}
    \caption{The dynamics of $\mu_j(\tau)$ (solid) and $\sfE^\tau_\tau[\mu_j]$(dash-dotted) in the five varicosity neighborhoods along the 1-D section of Figure \ref{fig:example_fibers} when the varicosity sites at $j = 1$ and $j = 5$ fire in-phase at 2 Hz for the initial transient (left) and several periods later (right). The dashed lines are $\sfE^\freq_\infty[\nu_j]$, the lower bound of the period-averaged steady-state. Due to the symmetry of the arrangement, the curves for $j=1$ and $j=2$ are precisely on the ones for $j=5$ and $j=4$, respectively. All other parameters are set to their default values (see Table~\ref{tab:parameters}).}
    \label{fig:2HzCompare}
\end{figure}

The difference of the true, numerically computed, period-averaged dynamics $\sfE^\freq_\tau[\mu_j]$ with the lower-bound period-averaged dynamics $\sfE^\freq_\tau[\nu_j]$ for in-phase 2 Hz firing of the left varicosity and right varicosity can be found in Figure \ref{fig:tave_diff}. We notice that, the lower the average serotonin amount around the particular varicosity, the tighter the lower-bound period-averaged dynamics are to $\sfE^\freq_\tau[\mu_j]$. Recall that this lower-bound was derived via Jensen's inequality (Proposition \ref{prop:period-ave}) to deal with the Michaelis-Menten nonlinearity. But when the serotonin amount is low, this nonlinearity can be well approximated by a linear function, which leads to a tighter match with Jensen's inequality as observed.

\begin{figure}[!hbtp]
    \centering
    \includegraphics[width=0.8\linewidth]{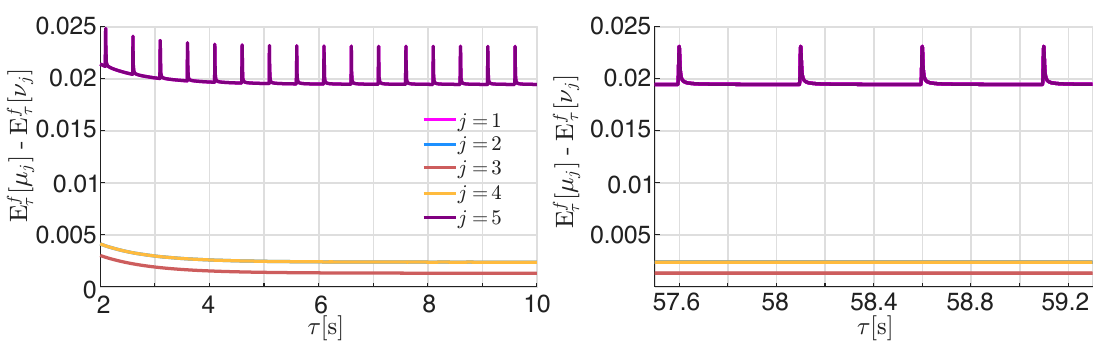}
    \caption{The difference of the true, though numerically computed, $\sfE^\freq_\tau[\mu_j]$ with the lower-bound period-averaged $\sfE^\freq_\tau[\nu_j]$ of the dynamics of total concentrations in the five varicosity neighborhoods along the 1-D section in 2-D space of Figure \ref{fig:example_fibers} when the left varicosity and the right varicosity fire in-phase at \SI{2}{\Hz}. Note that the oscillatory artifacts, caused by averaging over one period while using a discrete time-step of $\Delta t=0.03$ for $\sfE^\freq_\tau[\mu_j]$, have the same period as the firing. Due to the symmetry of the arrangement, the curves for $j=1$ and $j=2$ are precisely on the ones for $j=5$ and $j=4$, respectively. All other parameters are set to their default values (see Table~\ref{tab:parameters}).}
    \label{fig:tave_diff}
\end{figure}

This lower bound can be easily calculated from our system of equations. Due to the relatively tight agreement, we can then use it to quickly estimate the serotonin reservoirs that are built up over time for different firing frequencies. To report these reservoirs in physically meaningful units that can be compared with experimental measurements, we first convert the non-dimensional amounts $\sfE^\freq_t[v_j]$ produced by our 2-D model into three-dimensional concentrations. To do so, we envision our 2-D geometry as corresponding to the imaging of a brain slice of thickness $\Delta z= \,\SI{1}{\micro\meter}$, so that each disc of radius $R_0$ is interpreted as the cross-section of a cylinder of height $\Delta z$ and volume $\pi R_0^2\Delta z$. Assuming our estimate provides the projected count in this volume, the nM concentration is then obtained by multiplying $\sfE^\freq_t[v_j]$ by the conversion constant $c_c=\mu_c/(\pi R_0^2h)\approx 353.67$ nM. Table \ref{tab:lower-bound_concentrations} summarizes the resulting serotonin concentrations, rounded to the second decimal place, for various neurologically relevant frequencies mentioned in \cite{Zhang2025b}.

\begin{table}[!hbtp]
\footnotesize
\caption{Long-term lower bounds of the serotonin concentrations in {\rm{nM}} (i.e., $\sfE^\freq_\tau[v_j]\cdot c_c$ with dimensional concentration conversion constant $c_c:= \mu_c/(\pi R_0^2\Delta z)\approx 353.67$ \rm{nM}, where $\Delta z = \SI{1}{\micro\meter}$) after $\tau = \SI{60}{s}$ rounded to the second decimal place in the five varicosity neighborhoods depicted on the right in Figure \ref{fig:example_fibers}, with all other parameter values as reported in Table~\ref{tab:parameters}. The potential brain states associated with each frequency are based on \cite{Zhang2025b}. At the rates exceeding \SI{20}{\Hz}, the neurochemical profiles of neurons have not been definitively characterized.} \label{tab:lower-bound_concentrations}
\centering
\begin{tabular}{||l l l l l l||} 
 \hline
 Frequency [Hz] & $c_c\sfE^\freq_{\SI{60}{s}} [v_1]$ [nM] & $c_c\sfE^\freq_{\SI{60}{s}} [v_2]$ & $c_c\sfE^\freq_{\SI{60}{s}} [v_3]$ & $c_c\sfE^\freq_{\SI{60}{s}} [v_4]$ & $c_c\sfE^\freq_{\SI{60}{s}} [v_5]$ \\ [0.5ex] 
 \hline\hline
 $2$ (wakefulness or alertness) &103.04  &12.35  &6.82 &12.35 &103.04\\
 \hline
 $8$ (reward expectations) &432.15  &52.62  &29.01  &52.62  &432.15 \\
 \hline
 $16$ &884.49  &109.54  &60.39  &109.54  &884.49 \\
 \hline
 $25$ (sensory stimuli) &1398.18  &175.76  &96.99  &175.76  &1398.18 \\
 \hline
 $32$ (unexpected rewards) &1799.03  &228.27  &126.11  &228.27  &1799.03 \\
 \hline
 $64$ (sexual or social rewards) &3636.31  &474.31  &263.71  &474.31  &3636.31 \\  [1ex] 
\hline
\end{tabular}
\end{table}

We now turn our attention to the 
long-time estimate of the spike maxima outlined in Lemma \ref{lem:longterm-spikemax}, as it remains to address the tightness of this estimate to the true long-term spike maxima. Considering the same geometry and firing configuration used thus far, we can find that this estimate is reasonable for firing frequencies of 2 Hz, 16 Hz, and 64 Hz, but it performs best for 16 Hz (Figure \ref{fig:spikemax}). Specifically, one observes that after a short transient, this estimate appears to be especially accurate for the non-firing varicosity neighborhood at all frequencies, however it noticeably  underestimates the spiking in varicosity neighborhoods with 2 Hz firing, while overestimating it at 64 Hz firing. The values of the maxima after 60 s and their estimates are provided in Table~\ref{tab:spike-max_concentrations}.

\begin{figure}[!hbtp] 
    \centering
    \includegraphics[width=0.95\linewidth]{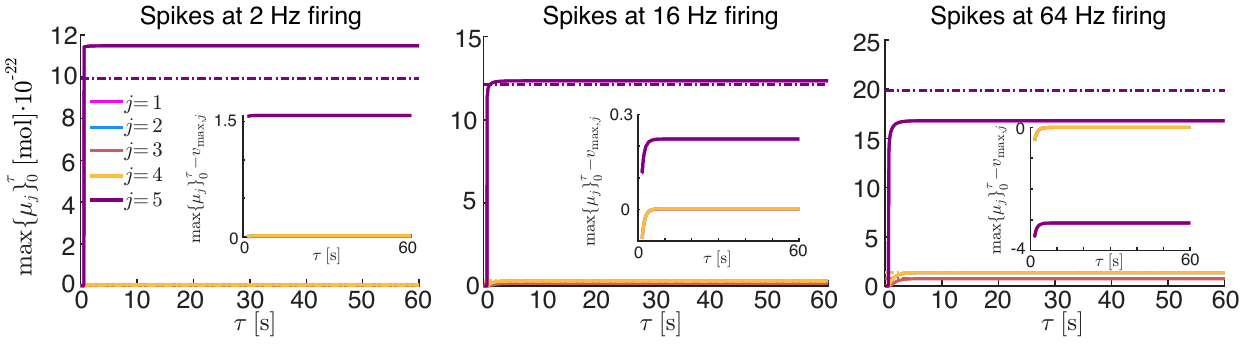}
    \caption{The running maxima of the serotonin amounts at all five varicosities along the section in Figure \ref{fig:example_fibers} (right) as solid lines and the corresponding spike maxima estimates in the same color as dash-dotted lines for spikes at \SI{2}{\Hz} (left), \SI{16}{\Hz} (center), and \SI{64}{\Hz} (right) respectively. The insets show the difference of the solid with the dashdotted lines of the left column, which is the absolute estimation error of the provided expression in Lemma \ref{lem:longterm-spikemax}. Due to the symmetry of the arrangement, the curves for $j=1$ and $j=2$ are precisely on the ones for $j=5$ and $j=4$, respectively. All other parameters are set to their default values (see Table~\ref{tab:parameters}).}    \label{fig:spikemax}
\end{figure}

\begin{table}[!hbtp]
\footnotesize
\caption{Spike maxima of the serotonin concentrations in {\rm{nM}} after \SI{60}{\s} (i.e., $\max_{0\leq \tau \leq 60 \,[s]}\{\mu_j\} \cdot c_c$, with dimensional concentration conversion constant $c_c:= \mu_c/(\pi R_0^2\Delta z)\approx 353.67$ \rm{nM}, where $\Delta z = \SI{1}{\micro\meter}$), along with their estimates from Lemma \ref{lem:longterm-spikemax} and estimation errors rounded to the second decimal place in the five varicosity neighborhoods depicted on the right in Figure \ref{fig:example_fibers}, with all other parameter values as reported in Table~\ref{tab:parameters}}. \label{tab:spike-max_concentrations}
\centering
\begin{tabular}{||l l l l l l||} 
 \hline
 Frequency [Hz]: variable [nM] & $j=1$ & $j=2$ & $j=3$ & $j=4$ & $j=5$ \\ [0.5ex] 
 \hline\hline
 2: $\max_{0\leq \tau \leq 60 \,[s]}\{u_j\} \cdot c_c$ &4062.86  &15.81  &7.48 &15.81 &4062.86 \\
 \hline
 2: $v_{\max,j}\cdot c_c$ & 3504.53  &12.35  &6.82 &12.35 & 3504.53 \\
 \hline
 2: $(\max_{0\leq \tau \leq 60 \,[s]}\{u_j\} - v_{\max,j})\cdot c_c$ & 558.33 &3.45  &0.66 &3.45 & 558.33 \\
\hline \hline
 16: $\max_{0\leq \tau \leq 60 \,[s]}\{u_j\} \cdot c_c$ &4364.73  &110.46  &60.89 &110.46 &4364.73 \\
 \hline
 16: $v_{\max,j}\cdot c_c$ & 4285.98  &109.54  &60.39 &109.54 & 4285.98  \\
 \hline
 16: $(\max_{0\leq \tau \leq 60 \,[s]}\{u_j\} - v_{\max,j})\cdot c_c$ & 78.75  &0.92  &0.50 &0.92 & 78.75 \\ 
\hline \hline
64: $\max_{0\leq \tau \leq 60 \,[s]}\{u_j\} \cdot c_c$ &5939.75  &474.51  &263.97 &474.51 &5939.75 \\
 \hline
 64: $v_{\max,j}\cdot c_c$ & 7037.80 &474.31  &263.71 &474.31 & 7037.80 \\
 \hline
 64: $(\max_{0\leq \tau \leq 60 \,[s]}\{u_j\} - v_{\max,j})\cdot c_c$ & -1098.04  &0.19  &0.26 &0.19 & -1098.04\\  [1ex]
 \hline
\end{tabular}
\end{table}

Since neurons naturally switch from one firing frequency to another when different stimuli are present, it is of interest to incorporate such switching into our investigation. Here, we show the dynamics of our system for the case where the firing occurs at initially 16 Hz, then switches to 64 Hz firing, and lastly to 2 Hz firing (Figure \ref{fig:switching16Hz64Hz2Hz}, solid line), along with the corresponding dynamics of the lower bound on the period-averaged dynamics, $\sfE^\freq_\tau[\nu_j]$ (dashed line). We note that care must be taken when computing the lower-bound period-averaged dynamics $\sfE^\freq_\tau[\nu_j]$ at the time steps around the switching event from one firing frequency to another. We observe that while the firing varicosities ($j = 1$ and 5) react immediately to the change in firing rates (as expected), there is a noticeable lag at the non-firing varicosities, with the serotonin ``reservoirs" taking a few seconds to accumulate.   
\begin{figure}[!hbtp]
    \centering
    \includegraphics[width=0.8\textwidth]{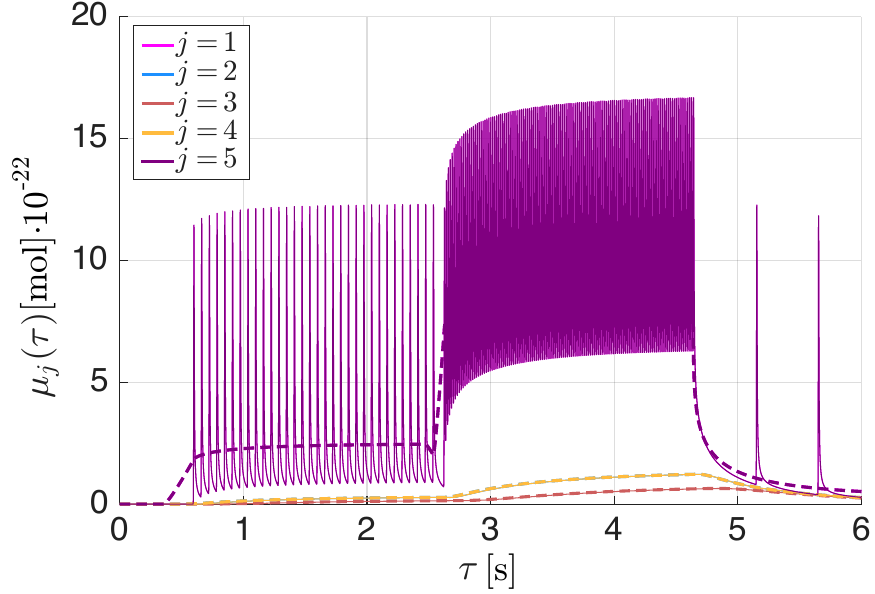}
    \caption{The dynamics of total concentrations in the five varicosity neighborhoods along the 1-D section in 2-D space of Figure \ref{fig:example_fibers} when the left varicosity ($\mu_1$) and the right varicosity ($\mu_5$) fire in-phase with initially \SI{16}{\Hz}, then with \SI{64}{\Hz}, and lastly with \SI{2}{\Hz}. One can observe that a steady-state of the period-averaged dynamics must exist due to the decay of the asymptote for \SI{64}{\Hz} firing when firing with \SI{2}{\Hz} after. The non-spiking dashed lines are the dynamics of the lower bound on the period-averaged steady-state, $\sfE^\freq_\tau[\nu_j]$. Due to the symmetry of the arrangement, the curves for $j=1$ and $j=2$ are precisely on the ones for $j=5$ and $j=4$, respectively. All other parameters are set to their default values (see Table~\ref{tab:parameters}).}
    \label{fig:switching16Hz64Hz2Hz}
\end{figure}

\subsection{Multiple rows of varicosities}\label{sec:multipleRows}
We now focus on changing the geometry of our problem by allowing for neighboring varicosities on each straightened fiber of Figure \ref{fig:example_fibers} (middle), which are now put in a parallel fiber setup. By iteratively adding a line of varicosities above, and then below, our single line, we inductively illustrate what quantitative changes of serotonin concentration are brought by additional varicosities in closer proximity.

We begin with a system of five fibers, each with two varicosities, with the leftmost and rightmost fibers firing periodically. In this situation, we have two lines of varicosities (Figure \ref{fig:2Hztwolines}, top row), and examine the dynamics of $\mu_j^{(2)}$ for a set of varicosities (the index $^{(2)}$ is added to denote the presence of two lines of varicosities). We also examine the steady-state of the lower-bound dynamics of Proposition \ref{prop:period-ave} to the extracellular serotonin (Figure \ref{fig:2Hztwolines}, top right). We find that the local dynamics in the neighborhood of each varicosity closely resemble those of the single-varicosity case at comparable points in time, with the main differences being a higher serotonin buildup and a distinct profile of extracellular serotonin concentration produced by the two rows of varicosities. This similarity carries through as well to the numerically estimated period-averaged dynamics $\sfE_\tau^\freq[\mu_j^{(2)}]$ (dash-dotted), and the steady-state of the lower-bound period-averaged dynamics $\sfE_\infty^\freq[\nu_j^{(2)}]$ (dashed), as shown in Figure~\ref{fig:twolines_Compare}.

\begin{figure}[!hbtp]
    \centering  
    \includegraphics[width=0.9\textwidth]{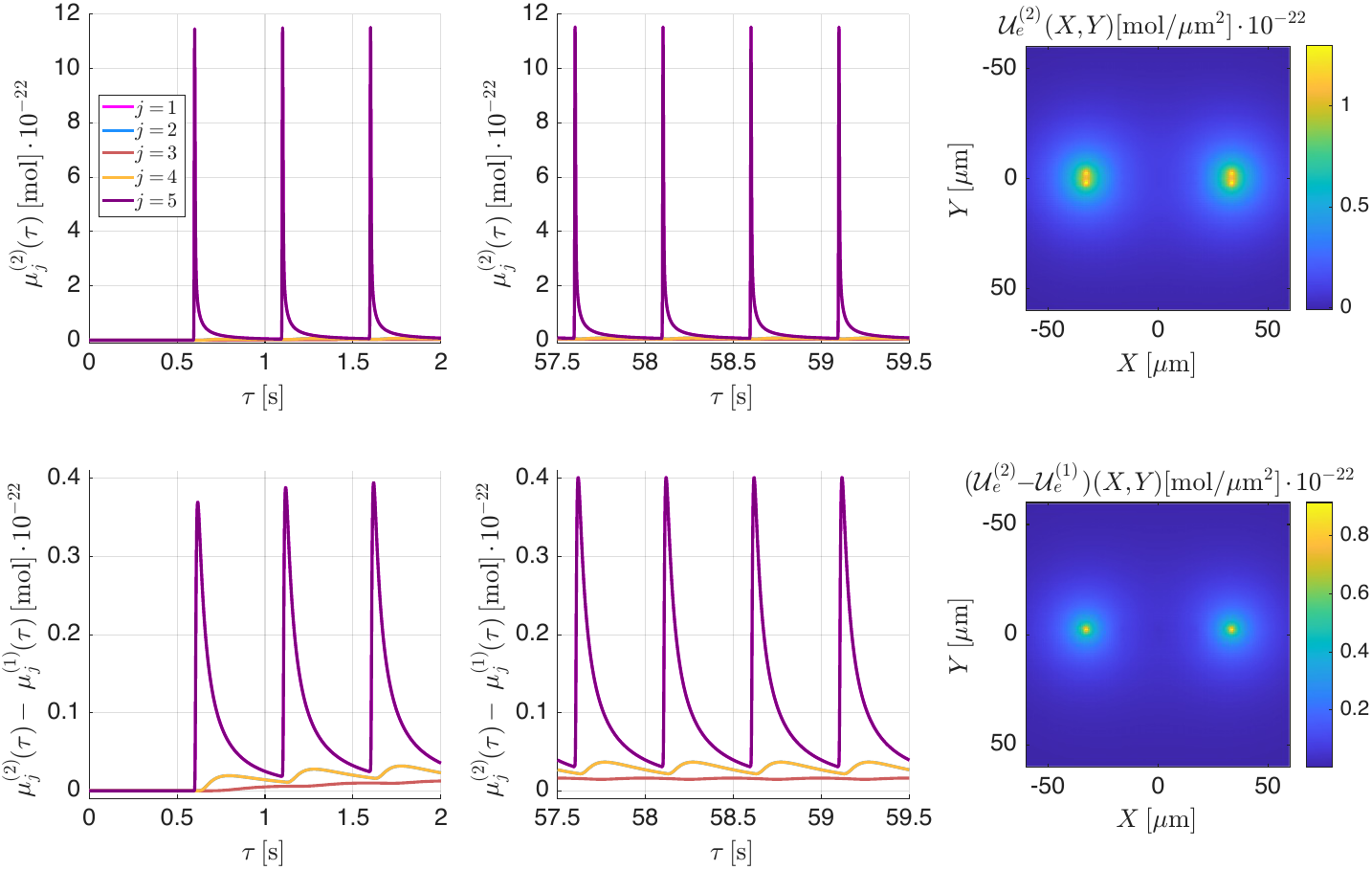}
    \caption{Top: The dimensional serotonin amount $\mu_j^{(2)}(\tau)$ for early (left column) and late (center column) time points. The right column shows the steady-state of the lower-bound period-averaged dynamics of Proposition \ref{prop:period-ave} to the extracellular serotonin concentration, $\mathcal{U}_e^{(2)}$. Bottom: Same as a the top row, except the difference between the single-varicosity, $\mu_j^{(1)}(\tau)$,  and two-varicosities, $\mu_j^{(2)}(\tau)$, setup along the bottom varicosity line of the five fibers. Firing frequency: \SI{2}{\Hz}. All other parameters are set to their default values (see Table~\ref{tab:parameters}). Due to the symmetry of the arrangement, the curves for $j=1$ and $j=2$ are precisely on the ones for $j=5$ and $j=4$, respectively.}
    \label{fig:2Hztwolines}
\end{figure}

To better highlight the difference between this setup of varicosities and our previous configuration, we next find the quantitative difference between the one-varicosity and two-varicosities setup on the five fibers (Figure \ref{fig:2Hztwolines}, bottom). In addition to the bigger buildup of serotonin, the spikes per firing event are also higher in the two-varicosity setup (Figure \ref{fig:2Hztwolines}, bottom, left and center panels). Further, we find significant spatial differences in the steady-state of the lower bound solution that originate, but spread beyond, at the location of the newly firing varicosity sites.

We iterate our exploration by adding an additional line of varicosities to consider three-varicosities setup on each of the five fibers (Figure \ref{fig:2Hzthreelines}). Again, we find the dynamics of $\mu_j^{(3)}$ to be quite similar to before, though the three firing varicosities are now quite visible in the steady-state of the lower-bound dynamics (Figure \ref{fig:2Hzthreelines}, top row). The same holds true for $\sfE_\tau^\freq[\mu_j^{(3)}]$ and $\sfE_\infty^\freq[\nu_j^{(3)}]$ (Figure \ref{fig:threelines_Compare}). Taking the difference of these time traces with the single-varicosity setup, along with a similar calculation for the steady-state plot, we find significant increase in transient peaks and long-time serotonin reservoirs across the domain, though most notable near the varicosities that are firing (Figure \ref{fig:2Hzthreelines}, bottom row).
\begin{figure}[!hbtp]
    \centering
    \includegraphics[width=0.9\textwidth]{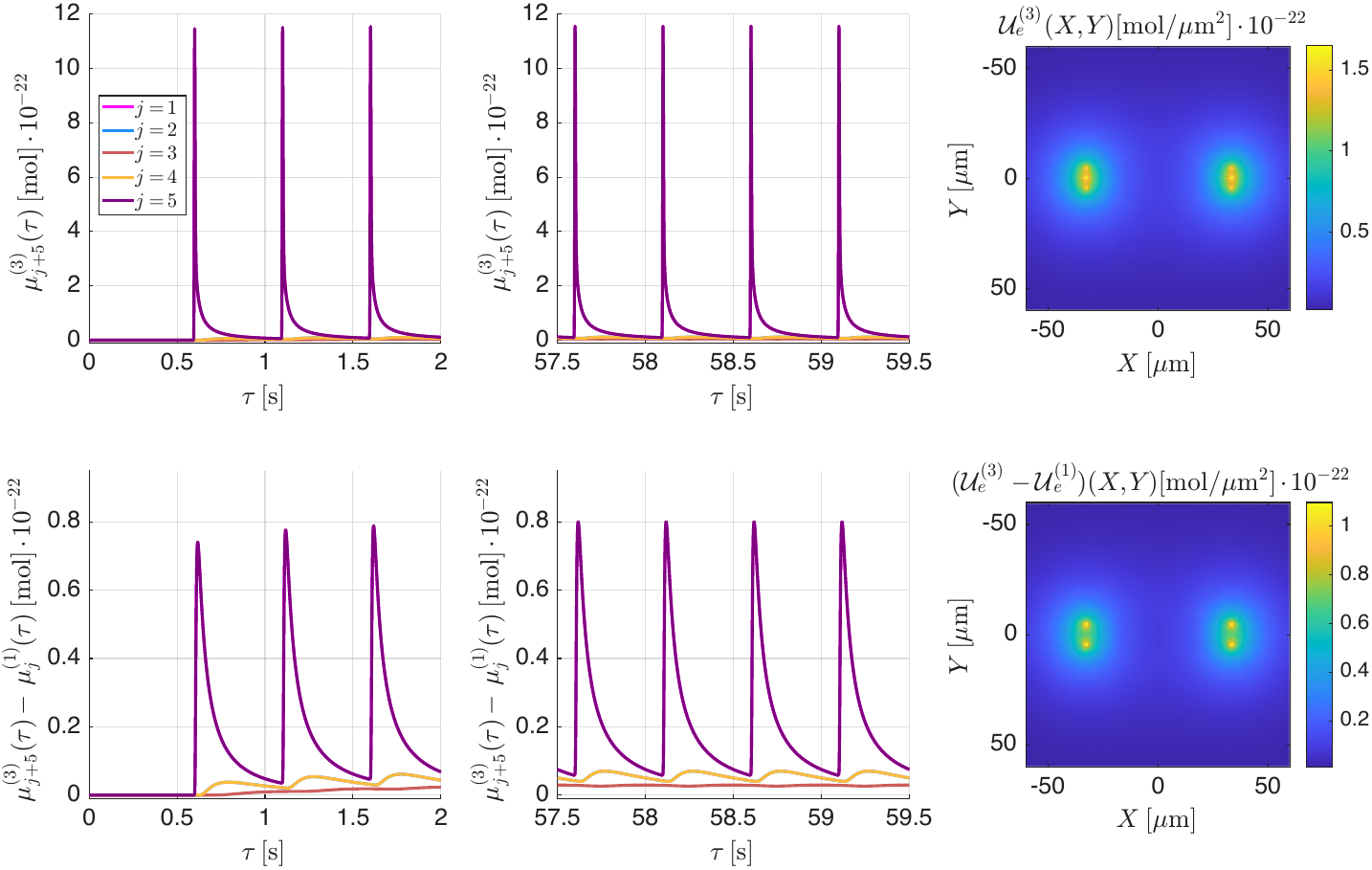}
    \caption{Same caption as Figure~\ref{fig:2Hztwolines}, but now comparing the middle horizontal line through varicosities of the setup with three varicosities per fiber with the single varicosity per fiber case.}
    \label{fig:2Hzthreelines} 
\end{figure}

\subsection{Comparison to linear local uptake} \label{sec:comparison_linear}
Thus far, we have used Michaelis–Menten uptake kinetics \cite{handy2021revising} to model serotonin uptake close to each varicosity. In this section we compare numerical results of the full time solutions and period-averaged steady-states for Michaelis-Menten uptake of the form 
\begin{eqnarray*}
    f_\text{uptake}(u_j) = -\frac{Vu_j}{K + u_j},
\end{eqnarray*}
with the ones for linear uptake of the form
\begin{eqnarray*}
    f_\text{L,uptake}(u_j) = -\frac{V}{K} u_j.
\end{eqnarray*}
Due to the loss of this nonlinearity, we note that our previously derived lower bound for the period-averaged dynamics, derived in Proposition \ref{prop:period-ave}, becomes exact. Further, we can solve for $\sfE^\freq_\infty[\vv_\text{L}] = \sfE^\freq_\infty[{\bf u}_\text{L}]$ analytically, yielding
\begin{align*}
    \sfE^\freq_\infty[{\bf u}_\text{L}] = \left( \frac{V}{K}\mathrm{M}+\gamma_\infty I\right)^{-1}r\freq\mathrm{M}\mathbbm{1}_\freq.
\end{align*}

For this comparison, we consider the setup with three varicosities per fiber. 
As can be seen from Figure \ref{fig:2Hzthreelines_linear} (top row), the period-averaged steady-states using $f_\text{L,uptake}$ are unsurprisingly similar to the ones when using $f_\text{uptake}$ (see Figure~\ref{fig:2Hzthreelines} (top row)). However, we clearly observe that the spikes are lower when using $f_\text{L,uptake}$, along with the fact that the period-averaged steady-state in the extracellular space and at the varicosities is always lower when using $f_\text{L,uptake}$ instead of the more realistic $f_\text{uptake}$ (Figure \ref{fig:2Hzthreelines_linear}, bottom row). 
\begin{figure}[!hbtp]
    \centering 
    \includegraphics[width=0.9\textwidth]{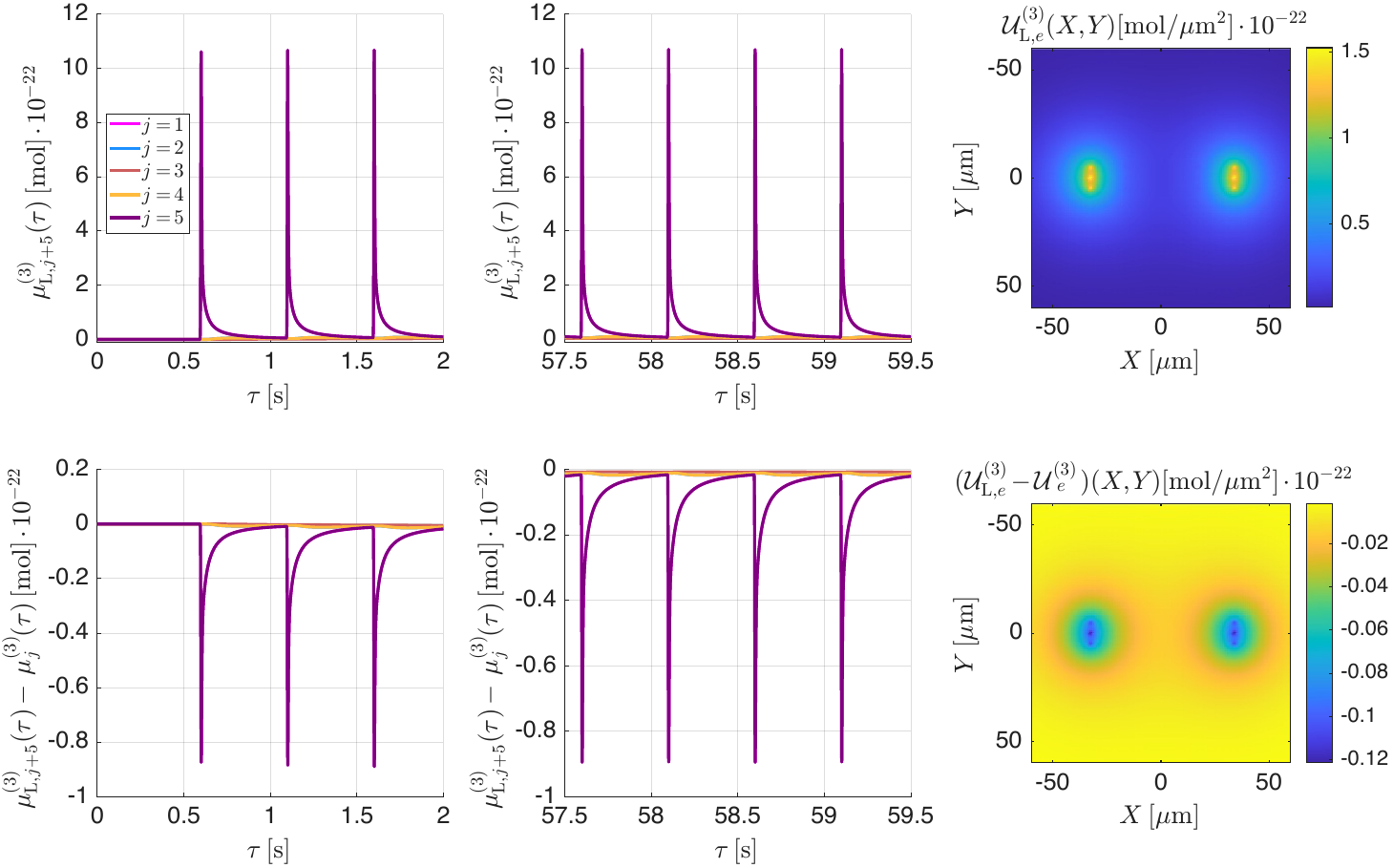}
    \caption{Top: The dimensional serotonin amount $\mu_{L,j}^{(3)}$ at \SI{2}{\Hz} for the linear local uptake model, (i.e., $f_{\text{L,uptake}}$), for early (left) and late (center) time points.. The right column shows the steady-state of the lower-bound period-averaged dynamics of Proposition \ref{prop:period-ave} to the extracellular serotonin concentration, $\mathcal{U}_e^{(3)}$, with the linear model, $\mathcal{U}_{L,e}^{(3)}$. Bottom: Same as a the top row, except the difference between the linear and nonlinear models along the middle varicosity line of the five fibers. Due to the symmetry of the arrangement, the curves for $j=1$ and $j=2$ are precisely on the ones for $j=5$ and $j=4$, respectively. All other parameters are set to their default values (see Table~\ref{tab:parameters}).}
    \label{fig:2Hzthreelines_linear}
\end{figure}

These results suggest that even at the low firing rate of 2 Hz, there is significant spillover of serotonin caused by accounting for the saturation of SERT. Indeed, even at low firing rates, the concentration at the varicosities during the spiking events is high enough to reliably saturate SERT leading to serotonin molecules escaping the nearby neighborhoods and adding to the serotonin reservoirs. This result is consistent with~\cite{handy2018}, which investigated this type of `spillover' at traditional neuronal synapses.

\section{Discussion} 
\label{sec:discussion}

This work introduces a compartmental-reaction diffusion framework that resolves serotonin dynamics around serotonergic varicosities while retaining analytical traction and computational efficiency. By capturing the solution to the full PDE–ODE formulation with an integro-ODE system with memory kernels, we reveal how spatial coupling through the extracellular space shapes local serotonin profiles across a network of release and uptake sites. The resulting theory explains how firing rate, varicosity density, and saturable uptake jointly generate extracellular reservoirs that depend on geometry and timing.

In this study, we further develop and extend results from~\cite{pelzward2025synchronized}, building upon one of their key results, adapted here in Proposition~\ref{prop:2ode}. Specifically, throughout this work we consider the limit of infinite permeabilities of the compartment boundaries in a well-defined way, so that, in contrast to the previous result, the compartments are boundary-less. This enables us to incorporate them as neighborhoods around the varicosities. Those neighborhoods then allow us to incorporate the dynamics near serotonin varicosities that are known biologically to date, without needing to specify precise permeabilities, which are also not yet fully known. Our theory additionally incorporates non-autonomous dynamics through the periodic (or stochastic) firing kicks, which necessitates the introduction of a period-averaging operator to investigate steady-state levels. This operator, not considered previously, required the development of new key results (Proposition~\ref{prop:period-ave} and Lemma~\ref{lem:longterm-spikemax}) to establish the existence and bounds of the system at steady-state firing. With the period-averaged system, we are able to leverage a similarly derived steady-state system and the efficient numerical time-marching scheme in~\cite{pelzward2025synchronized}, though in our specific context, namely infinite permeabilities. Hence, we have further developed the theory of~\cite{pelzward2025synchronized} to more realistically model and analyze the serotonin system at the level of individual varicosities.

The flexible mathematical framework developed in this study can support recent breakthroughs in experimental methods, some of which employ genetically-engineered probes to measure extracellular serotonin concentrations with high spatial precision~\cite{Zhang2025b}. It facilitates data interpretation and can predictively model the serotonin system’s spatial and temporal characteristics that cannot be directly or simultaneously captured by these experimental measurements. The need for advanced computational approaches is exemplified by the limited reliability of simplified calculations, despite the cutting-edge accuracy of raw data (e.g., arriving at only one neurotransmitter molecule being released from a synaptic vesicle~\cite{Zhang2025b}). The presented CRD-based framework, adapted specifically for microscale neurotransmitter dynamics, can bridge this gap.

The results obtained in this study support the hypothesis that serotonin neurotransmission can operate in the ``wired" and ``volume" modes \cite{Gianni2023a}. These modes are distinguished by the extent of serotonin diffusion from the release site. In ``wired'' neurotransmission, serotonin molecules released by a presynaptic site (varicosity) act on a single postsynaptic target, with no effect on adjacent synapses. Such neurotransmission is employed by many other neurotransmitter systems and can be captured by topological models (``A connects to B"), with no regard to the physical distances between synapses or brain regions~\cite{Lynn2019}. In ``volume" neurotransmission, serotonin can ``spill over" and diffuse away from its release site; as a result, a single presynaptic site can affect multiple postsynaptic sites in a physical neighborhood. In theory, the two modes can induce different computational regimes. However, they do not have to be mutually exclusive; a recent experimental study has suggested that the same neuron can be biased toward one or the other mode, depending on its firing rate~\cite{Zhang2025b}. Our computational analyses, based on a theoretical extension into the entire spatiotemporal continuum, lend strong support for this hypothesis. Specifically, lower firing rates may bias serotonergic neurons toward ``wired" neurotransmission and higher firing rates toward ``volume" neurotransmission, due to the buildup of a steady reservoir of serotonin in the ECS between fibers. Importantly, this transition is a smooth crossover governed by continuous parameters such as firing frequency and varicosity density, rather than a sharp bifurcation. This is consistent with the neuroscience literature, where the boundary between wired and volume transmission is not rigorously defined but instead depends on subjective cutoffs for the diffusing concentration and proximity to the postsynaptic target.

The proposed mathematical framework can tighten the estimates of the extracellular serotonin concentration in various brain regions. Despite decades of experimental research, these values vary over two or more orders of magnitude (also in overlapping regions studied by different researchers), with the cumulative range of approximately 0.1-65 nM \cite{Zhao2023}. These inconsistencies arise partially because of the limitations of invasive techniques that unavoidably damage serotonergic fibers in the sampling region and may also release blood platelets from damaged capillaries (all of which store serotonin at much higher concentrations). It is interesting to note that serotonin estimates in the blood plasma suffer from similar problems~\cite{Brand2011}. Given specific values or ranges of the key parameters, our model predicts the resultant extracellular serotonin concentration, including its spatiotemporal dynamics.   

Our examples used a regular grid of fibers and varicosities to isolate core mechanisms and to facilitate comparison across configurations. However, the model does not require this symmetry. The coupling kernels only depend on inter-site distances and the bulk operator, so the results extend directly to irregular or random varicosity placements \cite{Janusonis2020}. In practice, disordered configurations simply alter the entries of the distance-dependent matrices that govern the memory terms. Hence, the same marching solver and bounds carry over, and the framework can be used to explore disorder-induced variability in local concentrations and reservoir formation. While investigating these other configurations falls outside the scope of this study, they will likely lead to new insights. For example, the trajectories of serotonergic neurons can be modeled as paths of fractional Brownian motion, a stochastic process with memory~\cite{Janusonis2023,Janusonis2025}. In particular, this predicts an accumulation of serotonergic fibers at tissue borders, an observation consistent with experimental data. It may result in a correspondingly inhomogeneous distribution of extracellular serotonin, which can in turn affect the ingrowth of other serotonin fibers that may  sense local serotonin levels, in a (likely nonlinear) feedback loop~\cite{House2025}.

In this work, we considered a biologically realistic range of firing frequencies (2-64 Hz). In mouse brainstem slice preparations, serotonergic neurons in the dorsal raphe nucleus (DRN) produce a range of spontaneous firing rates, most of which fall within 0-4 Hz and show a near-normal statistical distribution at the population level~\cite{Mlinar2016a}. However, some neurons reach much higher firing rates (up to 8 Hz). This information is important for understanding locally sustained firing rates, but in brain slices serotonergic neurons no longer receive inputs from other brain regions (e.g., the major forebrain projections from the medial prefrontal cortex and the lateral habenula \cite{Soiza-Reilly2011}). 

Serotonergic firing patterns in the intact brain are likely to be more complex and depend on the interplay between individual neurons’ transcriptomes~\cite{Okaty2019,Ren2019} and the spike trains received from other neurons. Some serotonergic neurons show drastic changes in their firing rate in the wake-sleep transition and may fall silent in REM sleep~\cite{Jacobs1992,Brown2012}. Some neurons in the DRN of the rhesus macaque show reward-dependent modulation, and some DRN neurons in the cynomolgus macaque encode appetitive and aversive information, with firing rates reaching up to 40 Hz~\cite{Nakamura2008,Hayashi2015}. However, these studies have not identified the recorded neurons neurochemically (the DRN contains many non-serotonergic neurons). In the rat and mouse brains, identified serotonergic neurons have been shown to fire in the range of 0-20 Hz, often with higher rates achieved in phasic responses~\cite{Allers2003,Kocsis2006,Miyazaki2011,Liu2014,Cohen2015,Li2016,Okaty2019,Hashimoto2021}. The upper limit of the firing rate of serotonergic neurons remains poorly understood because the electrophysiological signatures of many serotonergic neuron subtypes described in single-cell RNA-seq studies remain uncharacterized~\cite{Okaty2019,Ren2019}. In this study, we used an extended firing range up to 64 Hz suggested by~\cite{Zhang2025b}. Even if firing rates at the high end of this range are not achieved in the normal brain, they might appear in abnormal or pharmacologically-induced states. With the selected parameter values, our simulations already show an emerging reservoir of extracellular serotonin at the 16 Hz-firing frequency.  

The asymptotics employed in this work rely on a small geometric parameter $\varepsilon$ that captures the ratio of varicosity neighborhood size to inter-varicosity spacing. The concentration conversion factor $c$ scales like $\mathcal{O}(1/\varepsilon)$ for our parameter values, while the theory requires it to be $\mathcal{O}(1)$. In our parameter regime $\varepsilon$ is small but not vanishing, which keeps $c$ numerically modest and the asymptotic approximations accurate in practice. This choice reflects a balance between biological realism and the need for controlled analytical reductions.

It is important to note that a number of advanced software platforms simulate neural dynamics at multiple scales, including NEURON \cite{hines1997neuron}, tools from the Open Brain Institute~\cite{OBI}, and The Virtual Brain \cite{sanz2015mathematical}. These environments increasingly support coarse-grained reaction–diffusion processes. However, the compartmental-reaction diffusion framework developed here targets a distinct regime. It resolves fine-scale, diffusion-mediated information transfer among many small, dynamically active compartments embedded in a continuous medium. Current large-scale platforms do not yet provide this capability at sub-micrometer resolution with memory-dependent coupling, so our approach offers complementary insight at this smaller scale.

The framework invites a few natural extensions. First, coupling the local compartments to simplified neuronal dynamics would close the loop between firing, release, uptake, and feedback modulation, enabling predictions about state-dependent switching between transmission modes. Second, replacing periodic firing with stochastic point processes will address physiologically relevant variability; the Poisson formulation in Appendix C provides a direct path. Third, our framework can be extended further to include heterogeneity in both varicosity sizes and positions, in order to investigate the impact this heterogeneity may have on serotonin concentration in the extracellular space. Lastly, while we formulated the model in two spatial dimensions, which matches common imaging projections and allows us to connect to recently developed mathematical theory, true brain tissue is three-dimensional, and diffusion in it is inherently three dimensional. Extending the strong localized perturbation approach to 3-D can be done in a similar fashion as the development of the 2-D asymptotic theory. However, the currently available 3-D theory tells us that the permeabilities need to scale differently in 3-D~\cite{pelz2024symmetry}, which will yield an asymptotic integro-ODE system of a slightly different form. Specifically, a 3-D extension would replace the 2-D heat kernel $G_{2D}(t, \vx, \vx_j) = \exp(-\sigma t)/t \cdot\exp(-|\vx - \vx_j|^2/(4Dt))$ and its logarithmic local behavior $B(t)\log|\vy| + C(t)$ with the corresponding 3-D Green function structure $G_{3D}(t, \vx, \vx_j) = \exp(-\sigma t)/(4\pi D t)^{3/2} \cdot\exp(-|\vx - \vx_j|^2/(4Dt))$ with local behavior $A(t)/|\vy| + B(t)$ and faster decay for $|\vx- \vx_j|\to\infty$. We expect that the faster spatial decay of concentration with distance in three dimensions will weaken coupling between varicosities; however, this effect must be evaluated in conjunction with the additional varicosities present in a 3-D volume, which may compensate for the reduced coupling. In parallel with the theoretical development, a corresponding numerical marching scheme must also be constructed and validated. Despite these hurdles, this extension is a natural next step that will enable direct comparison with volumetric experiments.

Lastly, the model can be imported with minor modifications to other neurotransmitter systems, such as dopamine and glutamate, where synaptic ``spillover" has been hypothesized under certain conditions~\cite{Shen2014,Walters2020a}. Importantly, it can also be used to model the effects of antidepressants, such as selective serotonin-reuptake inhibitors (SSRIs). While these pharmacological agents immediately increase extracellular serotonin levels, these effects do not produce direct therapeutic effects. Instead, relief appears to be caused by a slower and more complex rebalancing of neurotransmitter systems, leading to neuroplastic changes~\cite{Andrews2015,Witt2023}. Advanced computational approaches can support a deeper understanding of these clinically-relevant processes across several timescales.

\appendix
\appendixformat
\section{Numerical marching scheme}  \label{app:numerics}
    To numerically integrate \eqref{eq:reduced_infty} efficiently and accurately, the marching scheme developed in \cite{pelzward2025synchronized} is applied, along with its careful incorporation of the initial short-time behavior of the coupling functions $B_j$. For convenience, system \eqref{eq:reduced_infty} is reprinted here: 
    \begin{align*} 
    \frac{du_j}{dt} &= f(t, u_j) + B_j(t) \,, 
 \\
    \int_{0}^{t} B_j^{\prime}(s) E_1(\sigma(t-s))\, ds &= \eta_\infty B_{j}(t) + \gamma_\infty u_{j}(t) \nonumber \\
          & \qquad + \sum_{\stackrel{k=1}{k\neq j}}^{N} \int_{0}^{t} \frac{B_{k}(s) e^{-\sigma(t-s)}}{t-s} e^{-|\vx_j-\vx_k|^2/(4D(t-s))} \, ds\,.
  \end{align*}
  The memory-dependent integrals
  \begin{equation*}
      C_{jk}(t) := \int_{0}^{t} \frac{B_{k}(s) e^{-\sigma(t-s)}}{t-s} e^{-|\vx_j-\vx_k|^2/(4D(t-s))} \, ds \quad \text{and} \quad D_j(t) := \int_{0}^{t} B_j^{\prime}(s) E_1(\sigma(t-s))\, ds
  \end{equation*}
  pose the main difficulty for a numerical scheme. Not only is $E_1(z)$ singular at $z=0$ (though still integrable as a kernel) but the denominator in $C_{jk}$ additionally has a removable singularity. On top of that, at each time step, one needs to integrate over the full time interval until the current time point. Inspired by \cite{jiang2015efficient}, all these issues can be resolved by rewriting the kernels in the form
  \begin{equation*}
      E_1(\sigma t) = \sum_{l=-\infty}^\infty e_l\, e^{s_l t} \quad \text{and} \quad G_{jk}(t):= \frac{e^{-\sigma t}}{t} e^{-|\vx_j-\vx_k|^2/(4Dt)} = \sum_{l=-\infty}^\infty \omega_l\, e^{s_l t}
  \end{equation*}
  and by truncating the sums appropriately.
  
  Any well-behaved time- and space-dependent function $f$ with a Laplace transform $F$ and Bromwich contour $\Gamma_B$ for the inverse-Laplace transform can be written as the inverse-Laplace transform of its Laplace transform in the way
  \begin{eqnarray*}
      f(t, x) &=& \frac{1}{2\pi i} \int_{\Gamma_B} e^{st} F(s, x)\, ds \\
		&=& \frac{\chi}{2\pi} \int_{-\infty}^\infty e^{\chi P(z)t} \, F(\chi P(z), x) \,\cos(\alpha + i z)\, dz.
  \end{eqnarray*}
  In the last equality, the Bromwich contour $\Gamma_B$ has been deformed to a contour bending toward $\Re(s) \to - \infty$ away from $s = 0$ in the complex plane, defined by $s = \chi P(z) := \chi(1-\sin(\alpha+z))$ (see Figure \ref{fig:s_contour}) with parameters $\chi$ and $\alpha$ that can be optimally chosen. Sampling along this deformed contour, to obtain a quadrature for the integral, one only needs finitely many summands to approximate the integral to highest accuracy, due to the other summands corresponding to points along the deformed contour further toward $\Re(s) \to - \infty$ with fast exponential decay of the summands. 
  
  Hence, the approximation 
  \begin{equation*}
      f(t, x) \approx \sum_{l=-n}^n w_l e^{s_l t} =: f_n(t, x),
  \end{equation*}
  can be made with
  \begin{equation*}
    s_l = \chi (1 - \sin(\alpha + ilh))
    \quad \text{and} \quad  w_l = \frac{\chi h}{2 \pi} \cos(\alpha + ilh) \, F(s_l, x)
\end{equation*}
  and some discretization step $h$.
  The approximation error can then be obtained through
  \begin{equation*}
    \| f(t, x) - f_n(t, x)\| \leq \frac{\epsilon_f}{\sqrt{t}}, \quad t \in [\delta, T],
\end{equation*}
  as derived in \cite{jiang2015efficient}. 
  
    \begin{figure}
      \centering
      \includegraphics[width=.95\textwidth]{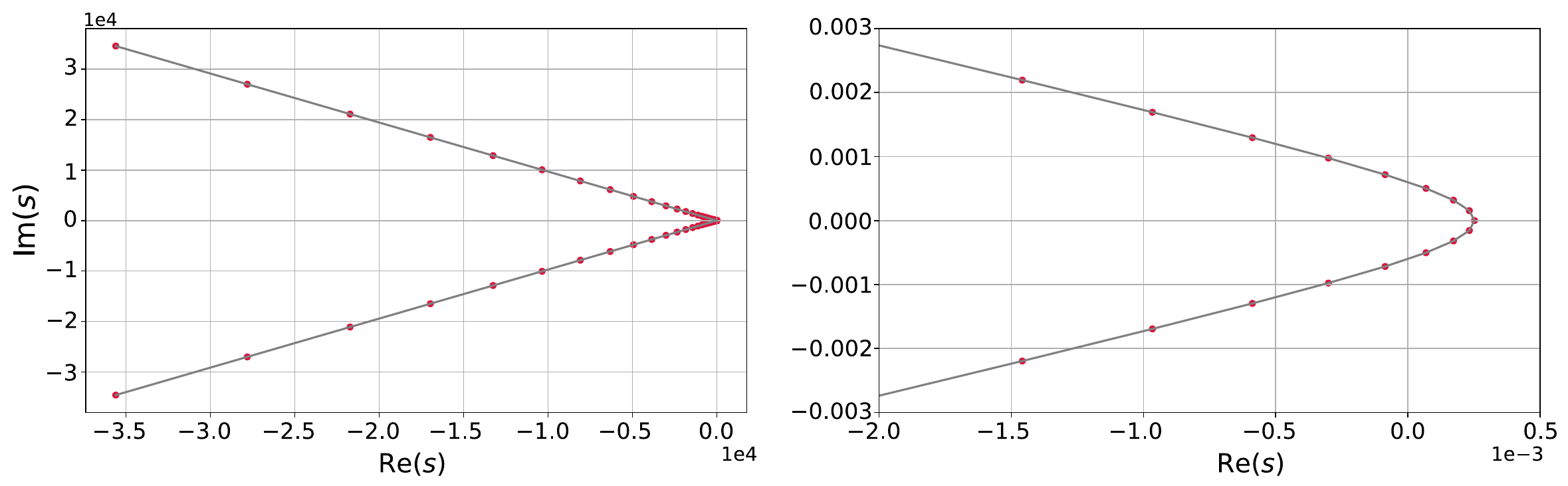}
      \caption{As in \cite{jiang2015efficient}, the deformed Bromwich contour $s = \chi(1-\sin(\alpha+z))$ with discretization points placed at $s_l = \chi (1 - \sin(\alpha+ilh))$ and equidistant step $h$ in the complex $s$-plane, zoomed out (left) and zoomed closer to $s=0$ (right). In this way, only a few $s_l$ are needed for a quadrature for the exponentially decaying summands as $\Re(s)\to-\infty$.  Parameters: $\chi = 0.00089$, $\alpha = 0.8$, and $h = 0.24706$.}
      \label{fig:s_contour}
  \end{figure}
  
  In this way, one can choose the accuracy $\epsilon_f$ and compute $n$ that leads to this accuracy at the start of the program before the numerical time stepping is initiated, in a matter of seconds. The parameters $\alpha$, $\chi$, $h$, $\delta$, and $T$ respectively correspond to the geometry of the deformed contour, its discretization, and the approximation time interval and can be optimized in order to minimize $n$.
  
  Further, we can also obtain approximations of the kernels $E_1(\sigma t)$ and $G_{jk}(t)$, each as a ``sum-of-exponentials". In addition, those sums enable us to develop a marching scheme without the need to integrate over all memory at each time-step, because a function of the form $\cF = \int_0^t g(\tilde{t})\,e^{s_l(t-\tilde{t})}\;d\tilde{t}$ with arbitrary but well-behaved $g$ allows for the local updating 
  \begin{equation*}
      \cF(t+\Delta t) = \cF(t) \,e^{s_l \Delta t} + e^{s_l\Delta t} \int_0^{\Delta t} g(t+z) \, e^{-s_l z} \; dz,
  \end{equation*}
  where $\Delta t$ is the time step of the equidistant discretization of the time domain.
  After incorporating the behavior 
  \begin{equation*}
   B_j(t) \sim -\frac{(u_{j}(0) + t u_j'(0)) \gamma_j}{\log\left({t/(\kappa^\infty_je^{-\gamma_e})}\right)} \, \quad \mbox{as} \quad t \to 0,
  \end{equation*}
  at time $t=\Delta t$, we use a simple Forward-Euler time-stepping for the $B_j$-part on the right-hand side of the ODE part of the integro-ODE system \eqref{eq:reduced_infty} and the Runge-Kutta 4 method for the remaining reaction kinetic terms of the ODE for $u_j$. The approximation of $u_j$ for the next time point is then used in the integro-ODE for $B_j$ using the mentioned marching scheme to obtain the numerical approximation for $B_j$ at the next time point. For more details, see \cite{pelzward2025synchronized}.  

\section{Derivation of the spike maxima estimate} \label{spikemaxDeriv}
The derivation of the long-term spike maxima estimate of Lemma \ref{lem:longterm-spikemax} goes as follows. Let $T_1 \gg 1$ and $T_2 \gg 1$ be large times with approximate release duration of a single spike $T_2 - T_1 = \Delta T := 4\theta$ in the interval $[T_1, T_2]$. To find an approximation for the spiking maxima, denoted by $\vv_{\max}$, we focus on the system for the lower-bound averaged dynamics $\sfE_t^\freq[v_j]$ because it can readily be solved for its steady-state $\sfE_\infty^\freq[v_j]$, $j\in\{1, ..., N\}$:
    \begin{subequations} 
    \begin{align}
        \frac{d\sfE^\freq_t[v_j]}{dt} &= -\frac{V\sfE^\freq_t[v_j]}{K + \sfE^\freq_t[v_j]} + r\freq(\mathbbm{1}_\freq)_j + \sfE^\freq_t[\tilde{B}_j] \,, \label{eq:v_1}
        \\
        \int_{0}^{t} \sfE^\freq_{s}[\tilde{B}_j'] \;E_1(\sigma(t-s))\, ds &= \eta_\infty \sfE^\freq_t[\tilde{B}_{j}] + \gamma_\infty \sfE^\freq_t[u_{j}] \nonumber
        \\ 
          & \qquad + \sum_{\stackrel{k=1}{k\neq j}}^{N} \int_{0}^{t} \sfE^\freq_{s}[\tilde{B}_{k}] \;\frac{e^{-\sigma (t-s)}}{t-s} e^{-|\vx_j-\vx_k|^2/(4D(t-s))} \, ds\,. \label{eq:v_2}
    \end{align}
    \end{subequations}
    In the middle of the interval $[T_1, T_2]$ we let one spike fire with distribution $r\delta_\theta$ in addition to the period-averaged firing $r\freq \mathbbm{1}_\freq$, where we choose the notation $v_{\max,j}$ and $\tilde{B}_{\max,j}$ for the dependent variables after this spike at $T_2$, instead of the abuse of notation $\sfE_{T_2}^\freq[v_j]$ and $\sfE_{T_2}^\freq[\tilde{B}_j]$. Integrating the first equation $\eqref{eq:v_1}$ from $T_1$ to $T_2$ over the four standard deviations $\Delta T$ of $r\delta_\theta$ gives 
    \begin{align*}
        v_{\max,j} - \sfE_{T_1}^\freq[v_j] &= -\int_{T_1}^{T_2}\frac{V\sfE^\freq_s[v_j]}{K + \sfE^\freq_s[v_j]}\; ds + r(\mathbbm{1}_\freq)_j\int_{T_1}^{T_2} \delta_\theta(s - \frac{T_2+T_1}{2})\;ds + r\freq \Delta T (\mathbbm{1}_\freq)_j \\
        &\;\;\;\; + \int_{T_1}^{T_2} \sfE_s^\freq[\tilde{B}_j]\;ds.
    \end{align*}
    We note that the second term on the right-hand side can be rewritten as 
    \begin{align*}
         \int_{T_1}^{T_2} \delta_\theta\left(s - \frac{T_2+T_1}{2}\right)\;ds &= \frac{1}{2\pi}\left[\int_{-\infty}^{\Delta T/(2\theta)}
         \exp\left(-\frac{u^2}{2}\right)\;du-\int_{-\infty}^{-\Delta T/(2\theta)}
         \exp\left(-\frac{u^2}{2}\right)\;du\right] \\
         &= \Phi\left(\frac{\Delta T}{2\theta}\right) - \Phi\left(-\frac{\Delta T}{2\theta}\right) \\
          &= 2\Phi\left(\frac{\Delta T}{2\theta}\right) -1 \\
         &= 2\Phi(2) -1,
    \end{align*}
    where $\Phi$ denotes the cumulative distribution function of the standard Normal Distribution. Further, since $T_1 \gg 1$, we can now approximate 
    \begin{align*}
        -\sfE_{T_1}^\freq[v_j] &\sim -\sfE_\infty^\freq[v_j] \quad\text{ as } \quad T_1 \to \infty,
    \end{align*}    
    and
    \begin{align*}
        -\int_{T_1}^{T_2}\frac{V\sfE^\freq_s[v_j]}{K + \sfE^\freq_s[v_j]}\; ds &\sim -\frac{V\sfE^\freq_\infty[v_j]}{K + \sfE^\freq_\infty[v_j]} \Delta T \quad\text{ as } \quad T_1 \to \infty. 
    \end{align*}
    Lastly, we can approximate the remaining integral by
    \begin{align*}
        \int_{T_1}^{T_2} \sfE_s^\freq[\tilde{B}_j]\;ds &\approx \frac{\Delta T}{n}\left(\frac{1}{2}\sfE_{s_0}^\freq[\tilde{B}_j]+\frac{1}{2}\sfE_{s_n}^\freq[\tilde{B}_j] + \sum_{k=1}^{n-1}\sfE_{s_k}^\freq[\tilde{B}_j]\right)  
        \\ &\approx \Delta T \left(\frac{1}{n}\sfE_{\infty}^\freq[\tilde{B}_j]+\frac{n-1}{n}\tilde{B}_{\max,j}\right) \\
        &= \left(\tilde{B}_{\max,j} - \sfE_{\infty}^\freq[\tilde{B}_j]\right) \frac{n-1}{n}\Delta T + \sfE_{\infty}^\freq[\tilde{B}_j]\Delta T.
    \end{align*}
   To see this, we first split the interval $[T_1, T_2]$ into $n$ equally spaced sub-intervals with endpoints $s_k$, $k\in\{0, ..., n\}$, and apply the trapezoidal rule. We then make the rough approximation of taking the interior points to be $\tilde{B}_{\max,j}$ and the two endpoints to be $\sfE_{\infty}^\freq[\tilde{B}_j]$. In this work, we choose $n=5$.
   \begin{remark}
       This choice of $n$ was found to give a close to optimal approximation error of our developed expression for the spike maxima for our firing frequency range $[2 \,\mathrm{Hz}, 32 \,\mathrm{Hz}]$, though other choices of $n$, such as $n=3$, which corresponds to Simpson's Rule, resulted in similar estimates. 
   \end{remark}
   
   In total, this leads to the following estimate for the spike maxima
    \begin{align} \label{eq:v_spike_todo}
        v_{\max,j} - \sfE_{\infty}^\freq[v_j] &= -\frac{V\sfE^\freq_\infty[v_j]}{K + \sfE^\freq_\infty[v_j]} \Delta T + (2\Phi(2) -1 + \freq \Delta T)\,r(\mathbbm{1}_\freq)_j + (\tilde{B}_{\max,j} - \sfE_{\infty}^\freq[\tilde{B}_j])\frac{4\Delta T}{5} \nonumber \\[0.5em]
        &\;\;\;\; + \sfE_{\infty}^\freq[\tilde{B}_j]\Delta T. 
    \end{align}
    Next we focus on obtaining an approximation for the difference $\tilde{B}_{\max,j} - \sfE_{\infty}^\freq[\tilde{B}_j]$ from the second equation $\eqref{eq:v_2}$. We note that without the additional single spike in $[T_1, T_2]$
    \begin{equation*}
        \int_0^{T_1} \sfE_s^\freq[B_j'] E_1(\sigma(T_2-s))\;ds \sim \int_0^{T_1} \sfE_s^\freq[B_j'] E_1(\sigma(T_1-s))\;ds - \int_{T_1}^{T_2} \sfE_s^\freq[B_j'] E_1(\sigma(T_2-s))\;ds,
    \end{equation*}
    since all integrals converge and $T_1, T_2\gg 1$. The latter integral can be approximated by zero because $\sfE_s^\freq[B_j'] = 0$ as $T_1\to\infty$. Also note, upon naming the kernel
    \begin{equation*}
        G_{jk}(t) := \frac{e^{-\sigma t}}{t} e^{-|\vx_j-\vx_k|^2/(4Dt)},
    \end{equation*}
    that, without the additional single spike in $[T_1, T_2]$,
    \begin{align*}
        \int_{0}^{T_1} \sfE^\freq_{s}[\tilde{B}_{k}] G_{jk}(T_2-s) \, ds 
        \sim \int_{0}^{T_1} \sfE^\freq_{s}[\tilde{B}_{k}] G_{jk}(T_1-s) \, ds - \int_{T_1}^{T_2} \sfE^\freq_{s}[\tilde{B}_{k}] G_{jk}(T_2-s) \, ds,
    \end{align*}
    since again all integrals converge and $T_1, T_2\gg 1$. The latter integral can be approximated by $\sfE_\infty^\freq[\tilde{B}_k]\int_{T_1}^{T_2} G_{jk}(T_2-s)\; ds$.
    
    All in all, subtracting \eqref{eq:v_2} at $T_1$ from $\eqref{eq:v_2}$ at $T_2$ including the additional single spike approximately leads to 
    \begin{align*}
            \int_{T_1}^{T_2} \sfE^\freq_{s}[\tilde{B}_j'] \;E_1(\sigma(T_2-s))\, ds = \eta_\infty (\tilde{B}_{\max,j}-\sfE_\infty^\freq[\tilde{B}_j]) + \gamma_\infty (v_{\max,j} - \sfE_\infty^\freq[v_j]) \nonumber \\
           \qquad + \sum_{\stackrel{k=1}{k\neq j}}^{N} \int_{T_1}^{T_2} \sfE^\freq_{s}[\tilde{B}_{k}] G_{jk}(T_2-s) \, ds - \sfE_\infty^\freq[\tilde{B}_k]\int_{T_1}^{T_2} G_{jk}(T_2-s)\;ds \,.
    \end{align*}
    All remaining integrals need to be approximated for us to arrive at an approximation for $\tilde{B}_{\max,j}-\sfE_\infty^\freq[\tilde{B}_j]$. We start with 
    \begin{eqnarray*}
        \int_{T_1}^{T_2} \sfE^\freq_{s}[\tilde{B}_j'] \;E_1(\sigma(T_2-s))\, ds &\approx& \frac{\tilde{B}_{\max,j} - \sfE_\infty^\freq[\tilde{B}_j]}{\Delta T} \int_{T_1}^{T_2} E_1(\sigma(T_2-s))\, ds \\
        &=& \frac{\tilde{B}_{\max,j} - \sfE_\infty^\freq[\tilde{B}_j]}{\Delta T}\big[s\,E_1(s)-e^{-s}\big]_0^{\sigma\Delta T} \\
        &=& (\tilde{B}_{\max,j} - \sfE_\infty^\freq[\tilde{B}_j])\left(E_1(\sigma\Delta T) + \frac{1 - e^{-\sigma\Delta T}}{\sigma\Delta T}\right),
    \end{eqnarray*}
    and then also obtain 
    \begin{eqnarray*}
        &&\int_{T_1}^{T_2} \sfE^\freq_{s}[\tilde{B}_{k}] G_{jk}(T_2-s) \, ds \\
        &\approx& \frac{\tilde{B}_{\max,k} + \sfE_\infty^\freq[\tilde{B}_k]}{2} \int_{T_1}^{T_2} G_{jk}(T_2-s)\;ds \\
        &=& \left(\frac{\tilde{B}_{\max,k} - \sfE_\infty^\freq[\tilde{B}_k]}{2} + \sfE_\infty^\freq[\tilde{B}_k]\right) \int_{T_1}^{T_2} G_{jk}(T_2-s)\;ds \\
        &=& \left(\frac{\tilde{B}_{\max,k} - \sfE_\infty^\freq[\tilde{B}_k]}{2} + \sfE_\infty^\freq[\tilde{B}_k]\right) \int_0^{\sigma\Delta T} \frac{e^{-s}}{s} e^{-|\vx_j-\vx_k|^2/(4Ds/\sigma)}\; ds \\
        &\approx& \left(\frac{\tilde{B}_{\max,k} - \sfE_\infty^\freq[\tilde{B}_k]}{2} + \sfE_\infty^\freq[\tilde{B}_k]\right) \int_0^{\sigma\Delta T} \frac{1-s}{s} e^{-|\vx_j-\vx_k|^2/(4Ds/\sigma)}\; ds \\
        &=& \left(\frac{\tilde{B}_{\max,k} - \sfE_\infty^\freq[\tilde{B}_k]}{2} + \sfE_\infty^\freq[\tilde{B}_k]\right) \left[E_1\left(\frac{|\vx_j-\vx_k|^2}{4Ds/\sigma}\right) + \frac{|\vx_j-\vx_k|^2}{4D/\sigma} E_1\left(\frac{|\vx_j-\vx_k|^2}{4Ds/\sigma}\right) - e^{-|\vx_j-\vx_k|^2/(4Ds/\sigma}) s\right]_0^{\sigma\Delta T} \\
        &=& \left(\frac{\tilde{B}_{\max,k} - \sfE_\infty^\freq[\tilde{B}_k]}{2} + \sfE_\infty^\freq[\tilde{B}_k]\right) \left(\left(1+\frac{|\vx_j-\vx_k|^2}{4D/\sigma} \right) E_1\left(\frac{|\vx_j-\vx_k|^2}{4D\Delta T}\right) - e^{-|\vx_j-\vx_k|^2/(4D\Delta T)}\sigma\Delta T\right) \\
        &=:& \left(\frac{\tilde{B}_{\max,k} - \sfE_\infty^\freq[\tilde{B}_k]}{2} + \sfE_\infty^\freq[\tilde{B}_k]\right) \mathcal{I}_{jk}(\Delta T).
    \end{eqnarray*}
    Hence, for $\tilde{B}_{\max,j} - \sfE_{\infty}^\freq[\tilde{B}_j]$, we arrive at
    \begin{align*}
            (\tilde{B}_{\max,j} - \sfE_\infty^\freq[\tilde{B}_j])\left(E_1(\sigma\Delta T) + \frac{1 - e^{-\sigma\Delta T}}{\sigma\Delta T}\right) &= \eta_\infty (\tilde{B}_{\max,j}-\sfE_\infty^\freq[\tilde{B}_j]) + \gamma_\infty (v_{\max,j} - \sfE_\infty^\freq[v_j]) \nonumber \\
           &\qquad + \sum_{\stackrel{k=1}{k\neq j}}^{N} \frac{1}{2} \mathcal{I}_{jk}(\Delta T)(\tilde{B}_{\max,k} - \sfE_\infty^\freq[\tilde{B}_k]) \,,
    \end{align*}    
    which can be written in vector-notation as
    {\begin{equation*}
        \tilde{M} (\tilde{B}_{\max} - \sfE_\infty^\freq[\tilde{B}]) = \gamma_\infty (\vv_{\max} - \sfE_\infty^\freq[\vv]),
    \end{equation*}
    where the matrix $\tilde{M}$ has entries
    \begin{equation*}
        \tilde{M}_{jj} = E_1(\sigma\Delta T) + \frac{1 - e^{-\sigma\Delta T}}{\sigma\Delta T} - \eta_\infty \quad \text{and} \quad \tilde{M}_{jk} = -\frac{1}{2} \mathcal{I}_{jk}(\Delta T), \quad j\neq k.
    \end{equation*}
    Using this estimate for the difference $\tilde{B}_{\max} - \sfE_{\infty}^\freq[\tilde{B}]$, \eqref{eq:v_spike_todo} becomes, in vector-notation,
    \begin{align}
        \vv_{\max} - \sfE_{\infty}^\freq[\vv] &= \left(I - \frac{4\Delta T}{5} \gamma_\infty \tilde{M}^{-1}\right)^{-1}\left(-V\sfE^\freq_\infty[\vv]\oslash\left(K\mathbbm{1} + \sfE^\freq_\infty[\vv]\right) \Delta T \right.\\
        &\hspace{0.5cm} \left.+ (2\Phi(2) -1 + \freq \Delta T)\,r\mathbbm{1}_\freq + \sfE_\infty^\freq[\tilde{\vB}]\Delta T\right).\nonumber
    \end{align}
    with all matrices generically being invertible. 
\section{Poisson-firing} \label{poissonSec}
Throughout this work, we made the simplifying assumption that the fibers fire periodically. However, the results easily extend to a stochastic setting where the firing times of the fibers are independent (or even dependent) and exponentially distributed, $T_{l-1}^{(k)}-T_l^{(k)} \sim \text{Exp}(1/\freq)$. Given a realization of firing times, $T_l^{(k)}(\omega)$, Corollary \ref{cor:2ode} can be used to solve for a sample path $u_j$, which yields the following system
\begin{subequations} \label{poissSystem}
\begin{align} 
    \frac{du_j}{dt} &= -\frac{Vu_j}{K+u_j} + r \sum_l \delta_\theta(t-k_RT_l^{(k)}(\omega)) + B_j(t) \,, \\
    \int_{0}^{t} B_j^{\prime}(s) E_1(\sigma(t-s))\, ds &= \eta_\infty B_{j}(t) + \gamma_\infty u_{j}(t) \nonumber \\
          & \qquad + \sum_{\stackrel{k=1}{k\neq j}}^{N} \int_{0}^{t} \frac{B_{k}(s) e^{-\sigma(t-s)}}{t-s} e^{-|\vx_j-\vx_k|^2/(4D(t-s))} \, ds\,.
\end{align}
\end{subequations}
In this framework, both $u_j$ and $B_j$ are random variables, and we study the evolution of their expectations across realizations, which are governed by the equations
\begin{subequations} \label{eq:expectation}
\begin{align} 
    \frac{d\bE[u_j]}{dt} &= -V\bE\left[\frac{u_j}{K+u_j}\right] + r \freq + \bE[B_j(t)] \,, \label{eq:exA}\\
    \int_{0}^{t} \bE[B_j(s)]^{\prime} E_1(\sigma(t-s))\, ds &= \eta_\infty \bE[B_j(t)] + \gamma_\infty \bE[u_{j}(t)] \nonumber \\
          & \qquad + \sum_{\stackrel{k=1}{k\neq j}}^{N} \int_{0}^{t} \frac{\bE[B_{k}(s)] e^{-\sigma(t-s)}}{t-s} e^{-|\vx_j-\vx_k|^2/(4D(t-s))} \, ds\,,
\end{align}
\end{subequations}
where we used Campbell's theorem~\cite{Kingman1993} to compute
the expectation of the random sum driven by these Poisson processes, i.e., $\bE\left[\sum_l \delta_\theta(t-k_RT_l^{(k)}(\omega))\right]$, obtaining simply $r\freq$. Similar to before, this system is analytically intractable due to the nonlinearity in \eqref{eq:exA}, but again by applying Jensen's inequality we can derive the following system for $v_j$ and $\tilde{B}_j$,
\begin{subequations} \label{eq:exApprx}
\begin{align} 
    \frac{dv_j}{dt} &= -V\frac{v_j}{K+v_j} + r \freq + \tilde{B}_j(t) \,, \\
    \int_{0}^{t} \tilde{B}_j^{\prime}(s) E_1(\sigma(t-s))\, ds &= \eta_\infty \tilde{B}_j(t) + \gamma_\infty v_{j}(t) \nonumber \\
          & \qquad + \sum_{\stackrel{k=1}{k\neq j}}^{N} \int_{0}^{t} \frac{\tilde{B}_{k}(s) e^{-\sigma(t-s)}}{t-s} e^{-|\vx_j-\vx_k|^2/(4D(t-s))} \, ds\,,
\end{align}
\end{subequations}
where $v_j(t) \le \bE[u_j(t)]$. Assuming that $v_j(t) \rightarrow v_{j,\infty}$ and $\tilde{B}_j(t)\rightarrow \tilde{B}_{j,\infty}$ as $t\rightarrow \infty$ for each $j \in \{1,...,N\}$, we can apply the results of \cite{pelzward2025synchronized} directly to find that the steady-state of $v_j(t)$ is provided by the following nonlinear algebraic system 
\begin{subequations} \label{poissSys}
\begin{align}
    0 &= -V\frac{v_{j,\infty}}{K+v_{j,\infty}} + r \freq + \tilde{B}_{j,\infty} \,, \\
    0 &= \eta_\infty \tilde{B}_{j,\infty} + \gamma_\infty v_{j,\infty} + 2\sum_{\stackrel{k=1}{k\neq j}}^{N} \tilde{B}_{k,\infty}K_0\left(\sqrt{\frac{\sigma}{D}}|x_j-x_k|\right),
\end{align}
\end{subequations}
where $K_0$, again, denotes the modified Bessel function of the second kind. \ref{poissonSM} provides an example of System~\eqref{eq:PDEsystem} with inputs driven by a Poisson process (Figures~\ref{fig:Pois25Hz} and \ref{fig:Euj}).

\section*{Author Contributions}
All authors have made substantial intellectual contributions to the study conception, 
execution, and design of the work. All authors have read and approved the final manuscript.  
In addition, the following contributions occurred:  Conceptualization: M.P., S.J., and G.H.; 
Methodology: M.P. and G.H.; Formal analysis and investigation: M.P. and G.H.; 
Writing - original draft preparation, review and editing: M.P., S.J., and G.H.; Funding acquisition: S.J. and G.H.

\section*{Access to Code} The MATLAB code to reproduce the main figures of the paper, along with Julia code for running the numerical algorithm and corresponding simulations, is available in a publicly accessible Zenodo repository (\url{https://doi.org/10.5281/zenodo.20159127}).

\bibliographystyle{siamplain}
{\footnotesize \bibliography{references}}

\pagebreak

\renewcommand{\thefootnote}{\fnsymbol{footnote}}

\renewcommand{\bfseries}{\fontseries{bx}\selectfont}
\begin{center}
	{\Large \textbf{Supplemental Materials: Compartmental-reaction diffusion framework for microscale dynamics of extracellular serotonin in brain tissue} \par}
	\vspace{1em}
	
	{Merlin Pelz\footnotemark[1],
		Skirmantas Janu\v{s}onis\footnotemark[2], and
		Gregory Handy\footnotemark[1]
		\par}
\end{center}

\footnotetext[1]{Department of Mathematics, University of Minnesota, Minneapolis, MN, USA
	(mpelz@umn.edu, ghandy@umn.edu).}
\footnotetext[2]{Department of Psychological \& Brain Sciences, University of California Santa Barbara, Santa Barbara, CA, USA (janusonis@ucsb.edu).}

\setcounter{equation}{0}
\setcounter{figure}{0}
\setcounter{table}{0}
\setcounter{page}{1}
\setcounter{section}{0}
\makeatletter
\renewcommand{\thepage}{SM\arabic{page}}
\renewcommand{\theequation}{S\arabic{equation}}
\renewcommand{\thefigure}{SM\arabic{figure}}
\renewcommand{\thesection}{SM\arabic{section}}


\renewcommand{\appendixformat}{%
	\titleformat{\section}
	{\normalfont\Large\bfseries}
	{\thesection}
	{1em}{}
}
\appendixformat

\section{Supplemental figures}
This section contains supplemental figures for the main text.

\begin{figure}[!htbp] 
	\centering
	\includegraphics[width=0.85\textwidth]{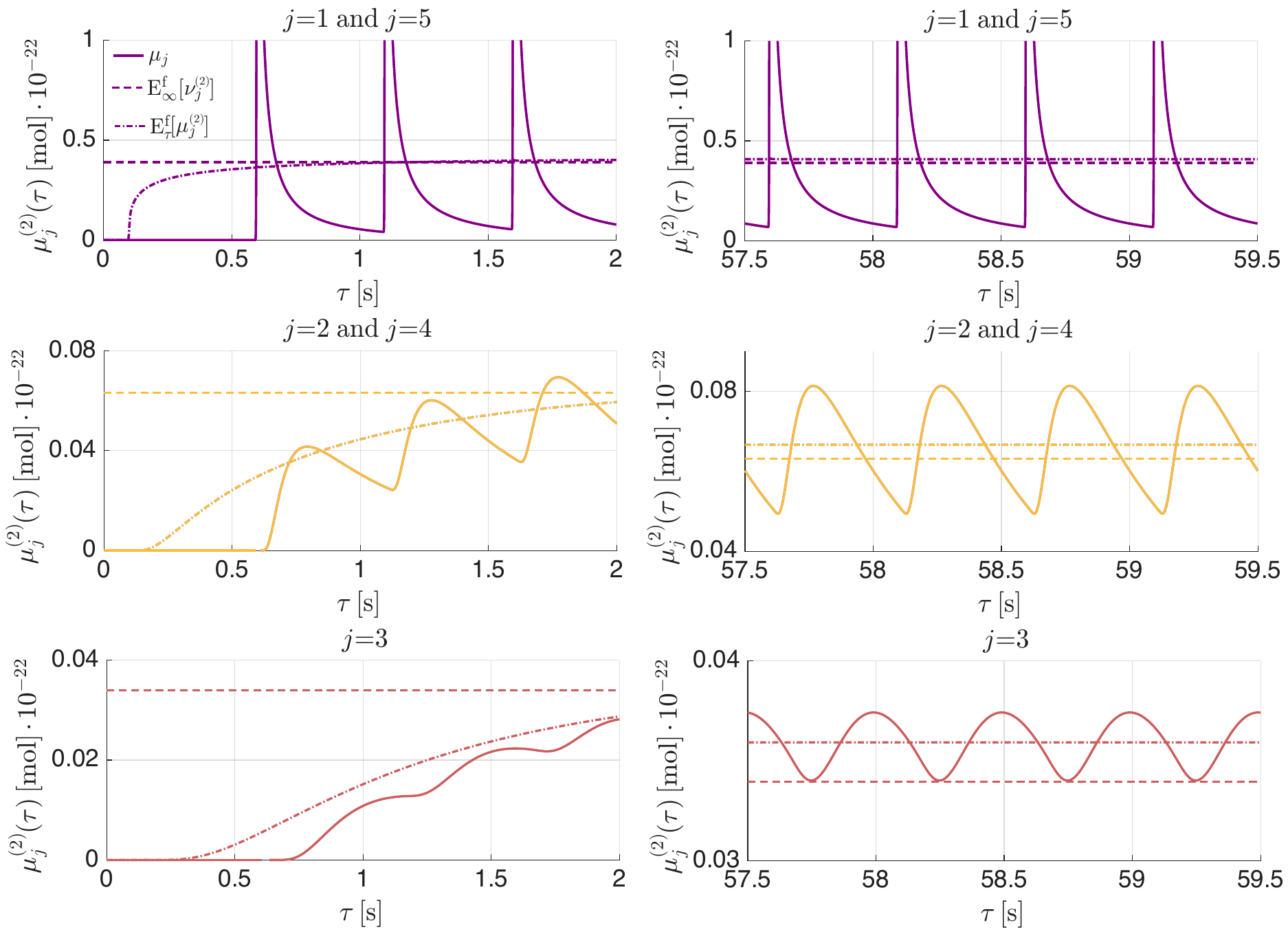}
	\caption{The dimensional serotonin amount $\mu_j^{(2)} (\tau)$ (solid line) along with its numerically estimated period-averaged dynamics $\sfE_\tau^\freq[\mu_j^{(2)}]$ as a dashed-dotted line and with the steady-state of the lower-bound period-averaged dynamics, $\sfE_\infty^\freq[\nu_j^{(2)}]$, as a dashed line for early (left) and late (right) time points. Here, we consider the same geometry as Figure~\ref{fig:2Hztwolines} (i.e., the system consists of five fibers, each with two varicosities, with the leftmost and rightmost fibers firing periodically at \SI{2}{\Hz}). All other parameters are set to their default values (see Table~\ref{tab:parameters}).}
	\label{fig:twolines_Compare}
\end{figure}

\begin{figure}[!htbp] 
	\centering
	\includegraphics[width=0.9\textwidth]{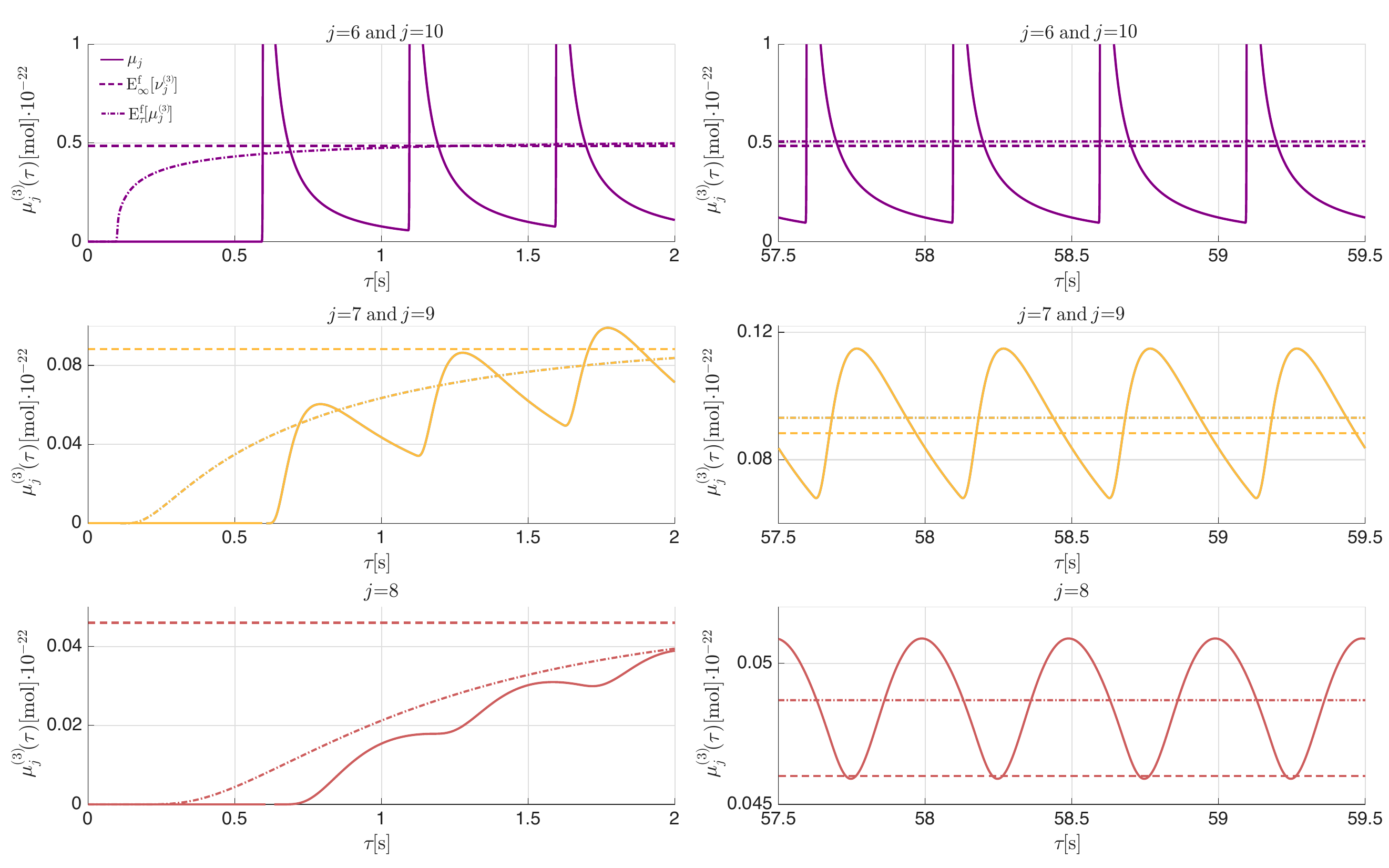}
	\caption{Same caption as Figure~\ref{fig:twolines_Compare}, but with the same geometry of Figure~\ref{fig:2Hzthreelines} (i.e., the system consists of five fibers, each with three varicosities, with the leftmost and rightmost fibers firing periodically at \SI{2}{\Hz}).}
	\label{fig:threelines_Compare}
\end{figure}

\section{Poisson firing example} \label{poissonSM}
Here we provide an example of System~\eqref{eq:PDEsystem} with inputs driven by a Poisson process, as described in Appendix \ref{poissonSec}. Figure \ref{fig:Pois25Hz} shows a realization of such dynamics with a Poisson process with an average firing rate of 25Hz.

\begin{figure}[!htbp] 
	\centering
	\includegraphics[width=0.95\textwidth]{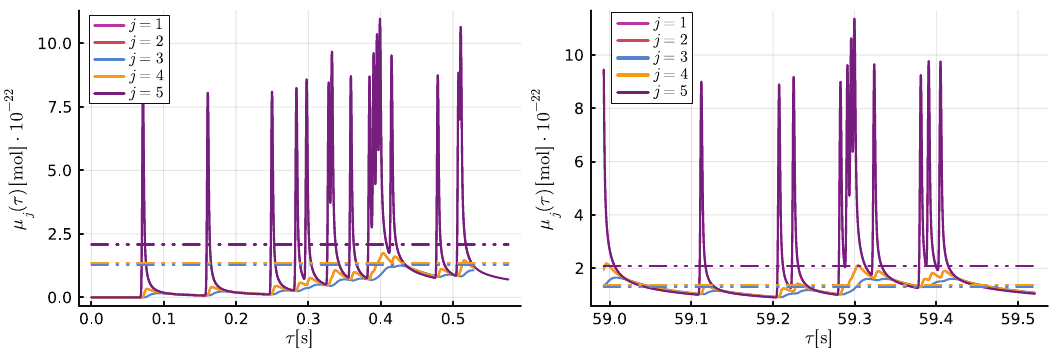}
	\caption{The dynamics of total serotonin amounts in the five varicosity neighborhoods along the 1-D section in 2-D space of Figure \ref{fig:example_fibers} when the left varicosity ($\mu_1$) and the right varicosity ($\mu_5$) fire in-phase distributed according the Poisson distribution with \SI{25}{\Hz} frequency in expectation. The dashed lines denote $\nu_{j,\infty}$, the lower bound on the expected steady-state across realizations derived from the nonlinear algebraic system \eqref{poissSys}. Due to the symmetry of the arrangement, the curves for $j=1$ and $j=2$ are precisely on the ones for $j=5$ and $j=4$, respectively. All other parameters are set to their default values (see Table~\ref{tab:parameters}).}
	\label{fig:Pois25Hz}
\end{figure}

The empirical approximation to the means of the $\mu_j$ with this Poisson firing (25 Hz) are depicted in Figure \ref{fig:Euj}. The sample means approach the expected steady-states of the lower-bound time-averaged kinetics $\nu_{j,\infty}$, $j\in\{1, 2, 3, 4, 5\}$, which are included in the plot as grey dashed lines.  

\begin{figure}[!htbp]
	\centering
	\includegraphics[width=0.95\linewidth]{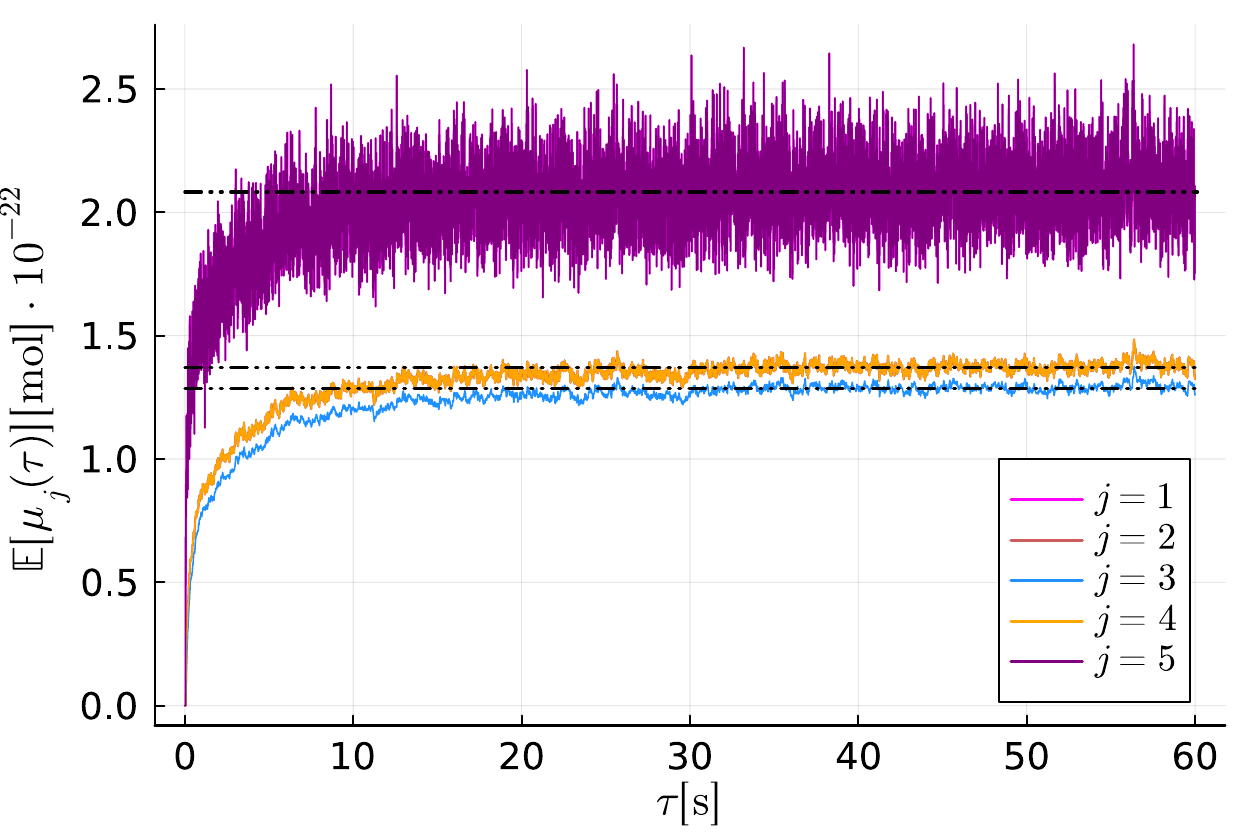}
	\caption{Approximation to $\bE[\mu_j]$ using 200 realizations computed from Equation \eqref{poissSystem} for the five-cell system along the 1-D section in $\bR^2$ shown in Figure \ref{fig:example_fibers} on the right. The dashed lines denote $\nu_{j,\infty}$, the lower bound on the expected steady-state across realizations derived from the nonlinear algebraic system \eqref{poissSys}. Note that the approximations at sites $j = 4$ and $5$ are equal to the ones of $2$ and $1$, respectively, and so only the values of $j \in \{3, 4, 5\}$ are visible. All other parameters are set to their default values (see Table~\ref{tab:parameters}).}
	\label{fig:Euj}
\end{figure}

\end{document}